\documentclass[11pt,journal,onecolumn]{IEEEtran}
\usepackage[T1]{fontenc}
\usepackage{lmodern,microtype,amsmath,amssymb,amsthm,mathtools,booktabs,array}
\usepackage{pgfplots}
\pgfplotsset{compat=1.18}
\usepackage{xcolor}

\usepackage[hidelinks]{hyperref}
\newtheorem{theorem}{Theorem}
\newtheorem{lemma}{Lemma}
\newtheorem{proposition}{Proposition}
\newcommand{\E}{\mathbb E}

\allowdisplaybreaks[2]
\hypersetup{pdftitle={A Sharp Fisher-Information Inequality for Boolean Pairs: Proof of the Amari--Kobayashi Conjecture},pdfsubject={A sharp Fisher-information inequality for arbitrary Boolean pairs}}
\title{A Sharp Fisher-Information Inequality for Boolean Pairs: Proof of the Amari--Kobayashi Conjecture}
\author{Jun Chen, Shun Watanabe, and Lei Yu\thanks{Authors are listed alphabetically.}
}
\begin{document}
\maketitle
\begin{abstract}
Consider two terminals that observe correlated uniform binary vectors and separately compress their observations to one bit. We prove that, for every blocklength and every correlation parameter in $(-1,1)$, the Fisher information retained by the resulting pair of bits is maximized by matching parity functions of an appropriate degree, thereby resolving a conjecture of Amari and Kobayashi. The proof combines a Fourier-based bound on the derivative of the normalized covariance with a scalar inequality for signed polynomials subject to a global range constraint.
\end{abstract}
\section{Introduction}\label{sec:intro}
How much information about a correlation parameter can two terminals preserve when each terminal is allowed to transmit only one bit? This question provides a concrete instance of statistical inference under distributed compression. The observations at either terminal may have a distribution that does not depend on the unknown parameter, even though their joint distribution does. In that setting, useful information is carried entirely by dependence, and the choice of the two encoders must preserve sensitivity to changes in that dependence.

The broader theory of statistical inference under multiterminal
data compression includes hypothesis testing and parameter
estimation under communication constraints; see Han and
Amari~\cite{hanamari}. Amari~\cite{amari} studied correlation
estimation for binary sources and compared compression schemes
through the Fisher information of their outputs. This formulation
leads to a concrete finite-dimensional question: for two correlated
uniform binary vectors, which pair of Boolean functions maximizes
the Fisher information about their correlation? The conjecture
of Amari and Kobayashi~\cite{amari2023,amari2024,kobayashi2023}
asserts that matching parity functions are optimal, with the
number of participating coordinates chosen according to the
correlation. 

This problem is closely related to the most informative Boolean
function problem. Kumar and Courtade~\cite{courtade} conjectured
that dictator functions maximize the mutual information between
a Boolean function of a uniform binary vector and the entire
vector observed through independent binary noise. This conjecture
was recently resolved in
\cite{chen2026proof,kytran2026,mahdavifar2026,openai2026contraction,kramer2026}.
For the two-function version, Pichler, Piantanida, and
Matz~\cite{pichler} had previously established the optimality of
dictator functions. Noise stability and related mutual-information
objectives have also been studied through Fourier
methods~\cite{li,odonnell}. Despite these connections, the
Fisher-information objective leads to a different optimization
problem: it measures local sensitivity to changes in the
correlation parameter. Consequently, a parity of several
coordinates can retain more Fisher information than a dictator
even though its output correlation has smaller magnitude.



The closest preceding result is due to Chen, Watanabe, and Yu~\cite{chen}. They established sharp Fisher-information comparisons for balanced pairs and for identical functions in the nonnegative-correlation regime. Their analysis also treats relative-entropy comparisons and Bayesian one-bit hypothesis testing. The present paper addresses the remaining Fisher-information question for arbitrary pairs, including unequal biases and negative correlations. We prove that matching parities attain the optimum for every blocklength and every correlation in $(-1,1)$. 

The main difficulty is the interaction between a polynomial derivative and a small output probability. Fourier analysis gives a bound on the derivative of normalized covariance that is sharp for a monomial. For balanced outputs, this estimate immediately yields the optimal Fisher information. For biased outputs, however, a joint cell can approach zero before normalized covariance reaches one. The reciprocal cell probability then introduces an additional curvature factor, which cannot be controlled by the balanced estimate alone.

Our proof uses the fact that the covariance polynomial must
remain feasible at every correlation, not only at the point
where its derivative is evaluated. We first factor the
four-cell curvature and bound it using only the feasible
boundary in the direction of the current covariance.
Together with the Fourier coefficient bound, this reduces
the theorem to a scalar inequality for signed polynomials
subject to a global upper bound. The central task is to
show that the derivative cannot be too large when the
polynomial is close to this bound. We establish this
control by combining spectral estimates with Taylor
estimates along suitable paths in the correlation
parameter. Near the upper boundary, these estimates force
the squared derivative to decrease with the distance to
the boundary, compensating for the singular curvature
factor. 


The rest of this paper is organized as follows.
Section~\ref{sec:main} states the main theorem and the two lemmas
used in its proof. Section~\ref{sec:curvature} proves the curvature
bound in  Lemma~\ref{lem:curvature}. Section~\ref{sec:scalar} proves the scalar inequality
in Lemma~\ref{lem:envelope}, separating negative slopes from positive
slopes and treating the latter according to the
correlation, the total logarithmic scale, and the
ceiling parameter.  Section~\ref{sec:conclusion}
concludes the paper. Throughout, all logarithms are
natural.

\section{Main Result and Supporting Lemmas}\label{sec:main}

Let $n$ be a positive integer, and let $X^n,Y^n$ be uniform random
vectors on $\{-1,1\}^n$ such that the coordinate pairs
$(X_i,Y_i)$ are independent and satisfy
$\E[X_iY_i]=r\in(-1,1)$ for every $i$. We call $Y^n$ an
$r$-correlated copy of $X^n$. For Boolean functions
$f,g:\{-1,1\}^n\to\{-1,1\}$, let $P^{f,g}_r$ denote the joint
law of $(f(X^n),g(Y^n))$. With $f$ and $g$ held fixed as $r$ varies,
define the Fisher information about $r$ retained by this
output pair as
\begin{align}\label{eq:Fisher}
I_{f,g}(r)
=
\sum_{\substack{x,y\in\{-1,1\}\\ P^{f,g}_r(x,y)>0}}
\frac{\bigl[\partial_r P^{f,g}_r(x,y)\bigr]^2}
     {P^{f,g}_r(x,y)}.
\end{align}
Restricting the sum to positive-probability outcomes accommodates constant functions. For $r\in(-1,1)$, every input pair has positive probability, so the output support is independent of $r$; if both functions are nonconstant, all four output probabilities are positive.



\subsection{The parity benchmark}\label{sec:Fisherbenchmark}

Let $S\subseteq\{1,\ldots,n\}$ be nonempty, with $|S|=k$, and
consider the matching parity functions
\begin{align}
f(x^n)=g(x^n)=\prod_{i\in S}x_i.
\end{align}
Both outputs are balanced, and their correlation is $r^k$.
Consequently, each equal-sign output pair has probability
$(1+r^k)/4$, whereas each opposite-sign output pair has
probability $(1-r^k)/4$:
\begin{align}
&P^{f,g}_r(-1,-1)
=P^{f,g}_r(1,1)=\frac{1+r^k}{4},
\label{eq:fourcell_1}\\
&P^{f,g}_r(-1,1)
=P^{f,g}_r(1,-1)=\frac{1-r^k}{4}.
\label{eq:fourcell_2}
\end{align}
Substituting \eqref{eq:fourcell_1} and \eqref{eq:fourcell_2}
into \eqref{eq:Fisher} gives
\begin{align}\label{eq:parity_information}
I_{f,g}(r)=\frac{k^2r^{2k-2}}{1-r^{2k}}.
\end{align}
At $r=0$, the Fisher information is $1$ for degree-one
parities and $0$ for all higher-degree parities, in agreement
with the continuous extension of
\eqref{eq:parity_information}.

The Amari--Kobayashi conjecture asserts that matching parity
functions of a suitable degree maximize the retained Fisher
information among all Boolean pairs. Our main result
establishes this conjecture.

\begin{theorem}\label{main}
For every positive integer $n$, every pair of Boolean functions
$f,g:\{-1,1\}^n\to\{-1,1\}$, and every $r\in(-1,1)$,
\begin{align}\label{target}
I_{f,g}(r)\le M_n(r),
\end{align}
where
\begin{align}\label{eq:parity_envelope}
M_n(r):=
\begin{cases}
\displaystyle
\max_{1\le k\le n}
\frac{k^2r^{2k-2}}{1-r^{2k}},
& 0<|r|<1,\\[6pt]
1, & r=0.
\end{cases}
\end{align}
For each $r$, equality is attained by matching parity
functions of a degree maximizing the benchmark.
\end{theorem}

\begin{figure}[t]
\centering
\includegraphics[width=\linewidth]{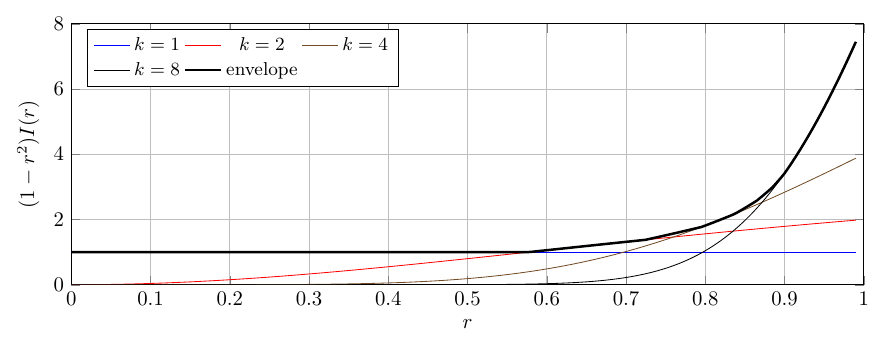}
\caption{Scaled Fisher information for matching parity
functions with $n=8$. Thin curves show selected degrees;
the thick curve shows the maximum over all eight degrees,
including those not individually displayed. The factor
$1-r^2$ removes the divergence as $r\uparrow1$.}
\label{fig:parity}
\end{figure}

Figure~\ref{fig:parity} illustrates how the optimal parity
degree changes with the correlation. For matching parities
of degree $k$,
\begin{align}\label{eq:scaled_parity_limit}
(1-r^2)I_{f,g}(r)
=
\frac{k^2r^{2k-2}}{\sum_{j=0}^{k-1}r^{2j}}
\longrightarrow k
\qquad\text{as }r\uparrow1.
\end{align}
Thus, for fixed $n$, degree $n$ is optimal sufficiently close
to $r=1$, whereas degree one is optimal at $r=0$ and in a
neighborhood of zero. Since the benchmark is even in $r$,
the same conclusions hold in terms of $|r|$. In particular,
when $n\ge2$, a comparison with dictators alone cannot yield
a sharp bound throughout $(-1,1)$.

\subsection{Normalization and Fourier estimates}
\label{sec:normalization}

This subsection prepares the reduction to the two lemmas
stated in Section \ref{sec:lemmas}. We express the output dependence
through a normalized covariance polynomial and establish
the range and coefficient constraints needed for the scalar
lemma. We also derive a spectral estimate that settles
several cases directly. Finally, we choose the sign of the
polynomial so that the two signs of output covariance can
be treated within a common framework.

If either function is constant, the output law is independent
of the correlation parameter, so its Fisher information is
zero. We therefore assume that both functions are nonconstant.
Replacing $g(y^n)$ by $g(-y^n)$ reflects the correlation
parameter and preserves Fisher information at the reflected
point. Since $M_n$ is even, it suffices to consider nonnegative
correlations.

Output relabeling and exchange of the two terminals allow
us to assume
\begin{align}\label{eq:pq}
0<p:=\Pr(f(X^n)=-1)
\le q:=\Pr(g(Y^n)=-1)\le\frac12.
\end{align}
These operations preserve Fisher information because they
only permute the terms in its defining sum. Specifically,
each output can first be complemented so that its probability
of being $-1$ is at most $1/2$, after which the terminals
can be exchanged if necessary.

Throughout this subsection, $v$ denotes a variable correlation
parameter, while $r$ denotes a correlation at which a
polynomial or its derivative is evaluated. The functions
$f$ and $g$ remain fixed as the correlation parameter varies.

\paragraph{The two endpoints}
We first express the dependence between the two outputs
by a single normalized covariance and determine the bounds
on this quantity imposed by the output marginals. 
These bounds supply the global range constraint used
in the scalar lemma and the endpoint parameters used
to express and bound the four-cell curvature.


Define
\begin{align}
\omega_p=\frac{p}{1-p},\qquad
\omega_q=\frac{q}{1-q},\qquad
\beta=\sqrt{p(1-p)q(1-q)},
\end{align}
and put
\begin{align}
a=\sqrt{\omega_p\omega_q},\qquad
b=\sqrt{\omega_p/\omega_q}.
\end{align}
The assumption \eqref{eq:pq} implies
$0<\omega_p\le\omega_q\le1$. Hence
$a/b=\omega_q\le1$ and $b^2=\omega_p/\omega_q\le1$, giving
\begin{align}
0<a\le b\le1.
\end{align}

For $v\in[0,1]$, define
\begin{align}\label{eq:normalized_covariance}
z(v)=\Pr_v(f(X^n)=-1,g(Y^n)=-1),\qquad
s(v)=\frac{z(v)-pq}{\beta}.
\end{align}
We refer to $s(v)$ as the normalized covariance; as will be seen, it is
the correlation coefficient of $f(X^n)$ and $g(Y^n)$.
At $v=1$, the joint input law is understood as $Y^n=X^n$.
We include this endpoint to express the global range
constraint on the covariance polynomial. Fisher information
is evaluated only for correlation parameters in $(-1,1)$.

The four output probabilities are
\begin{align}
&P^{f,g}_v(-1,-1)=z(v),\label{eq:cells_mm}\\
&P^{f,g}_v(-1,1)=p-z(v),\label{eq:cells_mp}\\
&P^{f,g}_v(1,-1)=q-z(v),\label{eq:cells_pm}\\
&P^{f,g}_v(1,1)=1-p-q+z(v).\label{eq:cells_pp}
\end{align}
Under \eqref{eq:pq}, their nonnegativity is equivalent to
$0\le z(v)\le p$: the other two constraints are automatic
because $q\ge p$ and $p+q\le1$. Subtracting $pq$, dividing
by $\beta$, and using
\begin{align}
\frac{pq}{\beta}=a,\qquad
\frac{p(1-q)}{\beta}=b
\end{align}
therefore yields
\begin{align}\label{feasible}
-a\le s(v)\le b,\qquad 0\le v\le1.
\end{align}
For $0\le v<1$, every input pair has positive probability.
Since both functions are nonconstant, all four output
probabilities are then positive, so both inequalities
in \eqref{feasible} are strict.

\paragraph{Fourier expansion and the coefficient budget}
We next show that the normalized covariance is a polynomial
in the input correlation whose absolute coefficient sum
is at most one. Together with the range constraint above,
this gives the polynomial assumptions required by the
scalar lemma. The same coefficient bound will also allow
us to control the derivative of the normalized covariance.

Write $[n]=\{1,\ldots,n\}$. For each $S\subseteq[n]$, define
the Fourier character
\begin{align}
\chi_S(x^n)=\prod_{i\in S}x_i,
\end{align}
with $\chi_{\varnothing}\equiv1$. Under the uniform measure
on $\{-1,1\}^n$, these characters satisfy
\begin{align}\label{eq:character_orthogonality}
\E[\chi_S(X^n)\chi_T(X^n)]
=
\begin{cases}
1,&S=T,\\
0,&S\ne T.
\end{cases}
\end{align}
Indeed, their product is the product of the coordinates
belonging to exactly one of $S$ and $T$. If $S\ne T$, this
product contains at least one independent coordinate with
mean zero. Since there are $2^n$ characters, they form an
orthonormal basis for the real-valued functions on
$\{-1,1\}^n$.

Consequently, $f$ has the Fourier expansion
\begin{align}\label{eq:fourier_expansion}
f(x^n)=\sum_{S\subseteq[n]}\widehat f(S)\chi_S(x^n),
\qquad
\widehat f(S)=\E[f(X^n)\chi_S(X^n)],
\end{align}
and similarly for $g$. In particular,
$\widehat f(\varnothing)=\E[f(X^n)]$. Parseval's identity
and the Boolean property $f^2\equiv1$ give
\begin{align}\label{eq:parseval_boolean}
\sum_{S\subseteq[n]}\widehat f(S)^2
=\E[f(X^n)^2]=1.
\end{align}
Removing the constant coefficient therefore yields
\begin{align}\label{eq:fourier_variance}
\sum_{S\ne\varnothing}\widehat f(S)^2
=1-\E[f(X^n)]^2
=\operatorname{Var}(f(X^n)),
\end{align}
with the analogous identity for $g$.

The correlation structure of $(X^n,Y^n)$ gives a corresponding
orthogonality identity across the two inputs. Independence
of the coordinate pairs implies
\begin{align}
\E_v[\chi_S(X^n)\chi_T(Y^n)]
&=
\prod_{i\in S\cap T}\E_v[X_iY_i]
\prod_{i\in S\setminus T}\E[X_i]
\prod_{i\in T\setminus S}\E[Y_i]\notag\\
&=
\begin{cases}
v^{|S|},&S=T,\\
0,&S\ne T,
\end{cases}
\label{eq:correlated_characters}
\end{align}
where $v^0$ is understood as $1$, including at $v=0$.
The second case follows because every coordinate has
mean zero, while the first uses $\E_v[X_iY_i]=v$.
Substituting the Fourier expansions of $f$ and $g$ and
applying \eqref{eq:correlated_characters}, we obtain
\begin{align}
\E_v[f(X^n)g(Y^n)]
&=
\sum_{S,T\subseteq[n]}
\widehat f(S)\widehat g(T)
\E_v[\chi_S(X^n)\chi_T(Y^n)]\notag\\
&=
\sum_{S\subseteq[n]}
\widehat f(S)\widehat g(S)v^{|S|}.
\label{eq:fourier_crossmoment}
\end{align}
The term indexed by $S=\varnothing$ is
$\E[f(X^n)]\E[g(Y^n)]$. Thus the covariance is represented
by the nonconstant Fourier terms.

To relate this covariance to $s(v)$, recall that
$p=\Pr(f(X^n)=-1)$ and $q=\Pr(g(Y^n)=-1)$. Hence
\begin{align}
\E[f(X^n)]=1-2p,\qquad
\E[g(Y^n)]=1-2q.
\end{align}
Using \eqref{eq:cells_mm}--\eqref{eq:cells_pp}, we also have
\begin{align}
\E_v[f(X^n)g(Y^n)]
=1-2p-2q+4z(v).
\end{align}
Subtracting the product of the marginal means gives
\begin{align}
\operatorname{Cov}_v(f(X^n),g(Y^n))
&=4(z(v)-pq)\notag\\
&=
\sum_{S\ne\varnothing}
\widehat f(S)\widehat g(S)v^{|S|}.
\label{eq:covariance_fourier}
\end{align}
Dividing \eqref{eq:covariance_fourier} by $4\beta$ and
grouping subsets according to their cardinality yields
\begin{align}\label{eq:normalized_coefficients}
s(v)=\sum_{k=1}^n c_kv^k,
\qquad
c_k=\frac{1}{4\beta}
\sum_{\substack{S\subseteq[n]\\|S|=k}}
\widehat f(S)\widehat g(S).
\end{align}
In particular, $s$ is a polynomial of degree at most $n$
with no constant term. Its coefficients depend only on
$f$ and $g$, since the marginal input distributions do
not depend on the correlation parameter.

We next bound the sum of the absolute values of these
coefficients. By the triangle inequality, Cauchy--Schwarz,
and \eqref{eq:fourier_variance},
\begin{align}
\sum_{k=1}^n|c_k|
&\le
\frac{1}{4\beta}
\sum_{S\ne\varnothing}
|\widehat f(S)\widehat g(S)|\notag\\
&\le
\frac{1}{4\beta}
\left(\sum_{S\ne\varnothing}\widehat f(S)^2\right)^{1/2}
\left(\sum_{S\ne\varnothing}\widehat g(S)^2\right)^{1/2}
\notag\\
&=
\frac{\sqrt{\operatorname{Var}(f(X^n))
                 \operatorname{Var}(g(Y^n))}}{4\beta}\notag\\
&=1.
\label{eq:coefficient_budget}
\end{align}
The last equality follows from
\begin{align}
\operatorname{Var}(f(X^n))=4p(1-p),\qquad
\operatorname{Var}(g(Y^n))=4q(1-q),
\end{align}
and the definition
$\beta=\sqrt{p(1-p)q(1-q)}$.

We have therefore established
\begin{align}\label{spectralnormal}
s(v)=\sum_{k=1}^n c_kv^k,\qquad
s(0)=0,\qquad
\sum_{k=1}^n|c_k|\le1.
\end{align}
For arbitrary Boolean pairs, the products
$\widehat f(S)\widehat g(S)$ can have either sign.
Accordingly, the coefficients $c_k$ need not be
nonnegative, and the absolute-value bound in
\eqref{spectralnormal} is essential.

\paragraph{The basic spectral inequality}
The coefficient budget yields a derivative bound expressed
in terms of the parity envelope. This estimate will settle the zero-output-covariance case
below and one of the boundary cases in the proof of the
main theorem in Section \ref{sec:reduction}.

Fix an evaluation point $r\in(0,1)$ and write
\begin{align}
L=\sum_{k=1}^n|c_k|\le1.
\end{align}
By the definition of the parity envelope, for every
$1\le k\le n$,
\begin{align}
\frac{k^2r^{2k-2}}{1-r^{2k}}\le M_n(r).
\end{align}
Since $1-r^{2k}>0$, taking square roots gives
\begin{align}\label{eq:degree_derivative_bound}
kr^{k-1}\le
\sqrt{M_n(r)}\sqrt{1-r^{2k}}.
\end{align}
Thus each monomial derivative is controlled by the same
parity benchmark.

Differentiating the polynomial in \eqref{spectralnormal},
evaluating at $v=r$, and applying the triangle inequality
and \eqref{eq:degree_derivative_bound}, we obtain
\begin{align}
|s'(r)|
&=\left|\sum_{k=1}^n kc_kr^{k-1}\right|\notag\\
&\le\sum_{k=1}^n |c_k|kr^{k-1}\notag\\
&\le\sqrt{M_n(r)}
\sum_{k=1}^n|c_k|\sqrt{1-r^{2k}}.
\label{eq:spectral_derivative_step}
\end{align}
Cauchy--Schwarz, applied to the vectors with entries
$\sqrt{|c_k|}$ and
$\sqrt{|c_k|}\sqrt{1-r^{2k}}$, gives
\begin{align}
\left(\sum_{k=1}^n|c_k|\sqrt{1-r^{2k}}\right)^2
&\le
\left(\sum_{k=1}^n|c_k|\right)
\left(\sum_{k=1}^n|c_k|(1-r^{2k})\right)\notag\\
&=L^2-L\sum_{k=1}^n|c_k|r^{2k}.
\label{eq:spectral_cs_slope}
\end{align}
To relate the last term to the polynomial value at $r$,
apply the triangle inequality and Cauchy--Schwarz once more:
\begin{align}
|s(r)|^2
&\le
\left(\sum_{k=1}^n|c_k|r^k\right)^2\notag\\
&\le
\left(\sum_{k=1}^n|c_k|\right)
\left(\sum_{k=1}^n|c_k|r^{2k}\right)\notag\\
&=L\sum_{k=1}^n|c_k|r^{2k}.
\label{eq:spectral_cs_value}
\end{align}
Combining \eqref{eq:spectral_derivative_step}--%
\eqref{eq:spectral_cs_value} yields
\begin{align}\label{eq:spectral_s}
s'(r)^2
\le M_n(r)\bigl(L^2-s(r)^2\bigr)
\le M_n(r)\bigl(1-s(r)^2\bigr).
\end{align}
No assumption on the signs of the coefficients is needed.
In fact, the argument applies to every real polynomial
with no constant term, degree at most $n$, and sum of
absolute coefficients at most one.

\paragraph{Zero covariance and zero correlation}
Before introducing the sign-dependent normalization, we
dispose of the cases in which the output covariance
vanishes. At such a point, Fisher information equals the
squared derivative of the normalized covariance, so the
preceding spectral estimate applies directly. We then
treat zero input correlation using the coefficient budget.

Suppose first that $r\in(0,1)$ and $s(r)=0$.
By \eqref{eq:normalized_covariance}, this is equivalent
to $z(r)=pq$. The four output probabilities therefore
factor into their marginals:
\begin{align}
&P^{f,g}_r(-1,-1)=pq,\\
&P^{f,g}_r(-1,1)=p(1-q),\\
&P^{f,g}_r(1,-1)=(1-p)q,\\
&P^{f,g}_r(1,1)=(1-p)(1-q).
\end{align}
Thus zero covariance is equivalent to independence for
this binary output pair.

The marginal probabilities $p$ and $q$, and hence $\beta$,
do not depend on the correlation parameter. Differentiating
the identity $z(v)=pq+\beta s(v)$ with respect to $v$ and
evaluating at $v=r$ gives
\begin{align}
z'(r)=\beta s'(r).
\end{align}
Differentiating the general four-cell formulas
\eqref{eq:cells_mm}--\eqref{eq:cells_pp} therefore gives,
at $v=r$,
\begin{align}
\beta s'(r),\quad -\beta s'(r),\quad
-\beta s'(r),\quad \beta s'(r),
\end{align}
respectively. Substituting these derivatives and the
output probabilities at $r$ into \eqref{eq:Fisher} gives
\begin{align}
I_{f,g}(r)
={}&\beta^2s'(r)^2
\left(
\frac{1}{pq}+\frac{1}{p(1-q)}
+\frac{1}{(1-p)q}+\frac{1}{(1-p)(1-q)}
\right).
\label{eq:information_zero_covariance}
\end{align}
The factor multiplying $s'(r)^2$ equals one, because
\begin{align}
&\beta^2
\left(
\frac{1}{pq}+\frac{1}{p(1-q)}
+\frac{1}{(1-p)q}+\frac{1}{(1-p)(1-q)}
\right)\notag\\
&=
\beta^2
\left(\frac1p+\frac1{1-p}\right)
\left(\frac1q+\frac1{1-q}\right)\notag\\
&=
\frac{\beta^2}{p(1-p)q(1-q)}\notag\\
&=1.
\label{independentcurvature}
\end{align}
It follows from \eqref{eq:spectral_s} that
\begin{align}
I_{f,g}(r)=s'(r)^2\le M_n(r).
\end{align}

At $r=0$, the inputs are independent, so the outputs are
also independent and the same calculation gives
$I_{f,g}(0)=s'(0)^2$. The polynomial representation
\eqref{spectralnormal} implies
\begin{align}
s(0)=0,\qquad s'(0)=c_1,\qquad
|c_1|\le\sum_{k=1}^n|c_k|\le1.
\end{align}
Consequently,
\begin{align}
I_{f,g}(0)=c_1^2\le1=M_n(0).
\end{align}
This settles both zero correlation and zero output
covariance.

\paragraph{Orienting the normalized covariance}
We now put the remaining cases into the form used by the
two lemmas in Section \ref{sec:lemmas}. By choosing a fixed sign
for the normalized covariance polynomial, we make its
value at the evaluation point positive and identify the
upper endpoint relevant to that sign. This allows positive
and negative output covariance to be treated by the same
information bound and scalar lemma.

Fix $r\in(0,1)$ with $s(r)\ne0$. We retain the full
polynomial $s(v)$ and its range constraint
$-a\le s(v)\le b$ for every $v\in[0,1]$.
Choose the sign of this polynomial according to its
value at $r$, so that the resulting polynomial has a
positive value at the evaluation point:
\begin{align}\label{eq:orientation}
(h,\lambda,\eta)=
\begin{cases}
(s,b,a),&s(r)>0,\\
(-s,a,b),&s(r)<0.
\end{cases}
\end{align}
Here $h=s$ or $h=-s$ is an identity of polynomials:
the sign is chosen once and held fixed as the argument
varies and when derivatives are taken. In general,
$h(v)$ is not $|s(v)|$.

If $s(r)>0$, the original interval becomes
$-\eta=-a\le h(v)\le b=\lambda$.
If $s(r)<0$, multiplying the original interval by $-1$
gives $-\eta=-b\le h(v)\le a=\lambda$.
Thus, in either case,
\begin{align}
-\eta\le h(v)\le\lambda
\qquad(0\le v\le1),
\end{align}
where $0<\lambda,\eta\le1$. Set
\begin{align}
d=h(r)=|s(r)|.
\end{align}
Since the bounds in \eqref{feasible} are strict at $v=r$,
we have
\begin{align}
0<d<\lambda\le1.
\end{align}

Writing $h(v)=\sum_{k=1}^n a_kv^k$, the coefficients
satisfy $a_k=c_k$ when $s(r)>0$ and $a_k=-c_k$ when
$s(r)<0$. Hence the sign choice preserves the coefficient
budget and the absence of a constant term:
\begin{align}\label{lemrange}
h(0)=0,\qquad
\sum_{k=1}^n|a_k|\le1,\qquad
h(v)\le\lambda\quad(0\le v\le1).
\end{align}
It also preserves the squared value and squared derivative,
so \eqref{eq:spectral_s} gives
\begin{align}\label{signedCS}
h'(r)^2\le M_n(r)\bigl(1-h(r)^2\bigr).
\end{align}
The two covariance signs have therefore been reduced to
the same setting: a signed polynomial with coefficient
budget at most one, a positive value $d$ at the evaluation
point, and a global upper bound $\lambda>d$.

\subsection{The information bound and the scalar envelope}\label{sec:lemmas}

We now state the two lemmas used to prove the main theorem.
The first bounds the Fisher information of a Boolean pair
in terms of the derivative and value of its normalized
covariance polynomial. The second shows that, under the
coefficient and range constraints established in
Section~\ref{sec:normalization}, this upper bound never
exceeds the parity benchmark.

In the first lemma,  we retain the normalization
of Section~\ref{sec:normalization}. In particular, the
evaluation point satisfies $r\in(0,1)$ and $s(r)\ne0$,
the polynomial $h$ is obtained from $s$ by the fixed sign
choice in \eqref{eq:orientation}, and
\begin{align}
d=h(r),\qquad 0<d<\lambda\le1.
\end{align}
The polynomial $h$ satisfies \eqref{lemrange}: it has no
constant term, its absolute coefficient sum is at most
one, and $h(v)\le\lambda$ throughout $[0,1]$.

For $0<d<\lambda\le1$, define
\begin{align}\label{defPhi}
\Phi_\lambda(d)
=
\frac{\lambda(1-d)}
{(1+d)(\lambda-d)(1-\lambda d)}.
\end{align}
Every factor in the denominator is positive on this
domain. This function bounds the factor that converts
the squared derivative of normalized covariance into
Fisher information.

\begin{lemma}[Curvature bound]\label{lem:curvature}
For the normalized Boolean pair described above,
\begin{align}\label{lemcurve}
I_{f,g}(r)\le h'(r)^2\Phi_\lambda(d).
\end{align}
\end{lemma}

To explain the form of this bound, recall that the output
marginals do not depend on the correlation parameter.
Consequently, all four output probabilities vary through
the single quantity $z(v)=pq+\beta s(v)$. Their derivatives
at $v=r$ have magnitude $\beta|h'(r)|$, so the Fisher
information is $h'(r)^2$ multiplied by a factor determined
by the four output probabilities. Lemma~\ref{lem:curvature}
bounds this factor using only the current normalized
covariance $d$ and its upper endpoint $\lambda$.
The proof in Section~\ref{sec:curvature} shows that this
bound holds for arbitrary, possibly unequal, output
biases.





The next lemma builds a bridge between the bound in Lemma \ref{lem:curvature} and the desired bound
$M_n(r)$.

\begin{lemma}[Scalar envelope]\label{lem:envelope}
Let $n$ be a positive integer, and let
\begin{align}
h(v)=\sum_{k=1}^n a_kv^k
\end{align}
have real coefficients satisfying
\begin{align}
\sum_{k=1}^n|a_k|\le1.
\end{align}
Suppose $0<\lambda<1$ and
\begin{align}
h(v)\le\lambda\qquad(0\le v\le1).
\end{align}
Fix $r\in(0,1)$ such that $d=h(r)\in(0,\lambda)$.
Then, we have 
\begin{align}\label{eq:scalar_minimum_bound}
h'(r)^2\Phi_\lambda(d)
\le M_n(r).
\end{align}
Equivalently, the same inequality holds after taking the
supremum over all admissible polynomials $h$. 
\end{lemma}

The global constraint in Lemma~\ref{lem:envelope} is essential:
it controls how rapidly $h$ can vary near a point where
its value is close to $\lambda$. The proof combines this
constraint with the coefficient budget to control $h'(r)^2\Phi_\lambda(d)$ in \eqref{eq:scalar_minimum_bound}.

Lemma~\ref{lem:curvature} concerns actual Boolean encoders.
Lemma~\ref{lem:envelope} is stated for a larger class of
signed polynomials, without requiring a Boolean
representation. This stronger formulation is sufficient
for our purpose: \eqref{lemrange} verifies its polynomial
hypotheses for the normalized Boolean pair.
The scalar lemma is proved
in Section~\ref{sec:scalar}.

\begin{figure}[t]
\centering
\includegraphics[width=\linewidth]{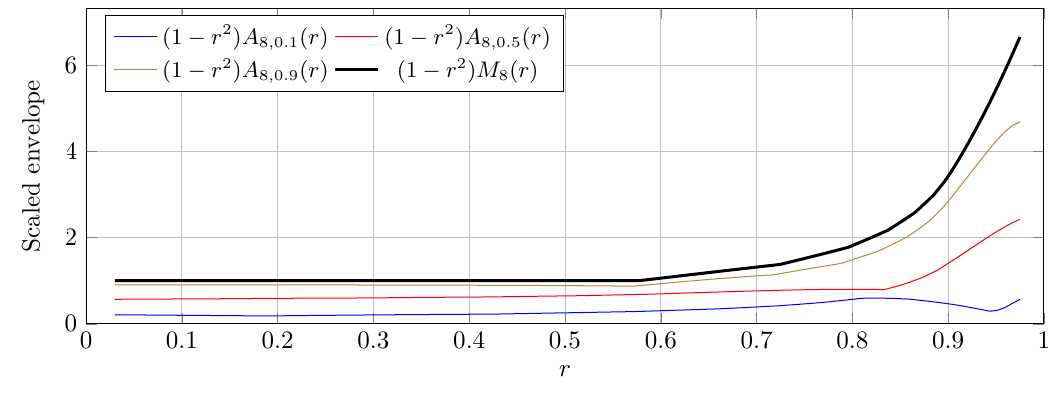}
\caption{Numerical illustration of the 
scalar envelope inequality for $n=8$ and
$\lambda=0.1,0.5,0.9$. The curves show
$(1-r^2)A_{8,\lambda}(r)$ and
$(1-r^2)M_8(r)$. The numerical envelopes
approximate the exact suprema but are not
certified globally optimal.}
\label{fig:scalar-envelope}
\end{figure}

Figure~\ref{fig:scalar-envelope} illustrates the 
scalar envelope inequality for $n=8$ and
$\lambda\in\{0.1,0.5,0.9\}$.
Let
\begin{align}
\mathcal H_{n,\lambda}
=
\left\{
h(v)=\sum_{k=1}^n a_kv^k:
\sum_{k=1}^n|a_k|\le1,\ 
h(v)\le\lambda\ \text{for all }v\in[0,1]
\right\}.
\end{align}
For $r\in(0,1)$, define
\begin{align}
A_{n,\lambda}(r)
=
\sup_{\substack{h\in\mathcal H_{n,\lambda}\\
                 0<h(r)<\lambda}}
h'(r)^2\Phi_\lambda(h(r)).
\end{align}
The curves corresponding to $A_{n,\lambda}(r)$ are
obtained by numerical optimization over the polynomial
coefficients. All functions are multiplied by $1-r^2$
to facilitate comparison. The computed values lie
below the parity benchmark at the sampled correlations,
consistently with Lemma~\ref{lem:envelope}. Among the three values of
$\lambda$ considered, the computed curve for
$\lambda=0.9$ lies closest to the parity benchmark
over the plotted range. These computations are
illustrative and are not used in the proof; they
do not establish monotonicity in $\lambda$ or
convergence to $M_8(r)$ as $\lambda\uparrow1$.
The changes in the slope of $M_8(r)$ correspond to
transitions in the degree attaining the maximum
in its definition.


\subsection{Proof of Theorem~\ref{main} from the lemmas}\label{sec:reduction}

The reductions and estimates in
Section~\ref{sec:normalization} already establish the
theorem when either encoder is constant, when the
correlation is zero, or when the output covariance
vanishes. Reflection of the second input reduces
negative correlations to positive ones, with the same
Fisher information at the reflected point and the same
parity benchmark. It therefore remains to consider
nonconstant encoders at a fixed correlation $r\in(0,1)$
with $s(r)\ne0$.

Apply the sign choice in \eqref{eq:orientation}, and put
$d=h(r)$. The resulting parameters satisfy
$0<d<\lambda\le1$. We distinguish the cases
$0<\lambda<1$ and $\lambda=1$.

Suppose first that $0<\lambda<1$. By \eqref{lemrange},
the polynomial $h$ has degree at most $n$, has no
constant term, and satisfies
\begin{align}
\sum_{k=1}^n|a_k|\le1,\qquad
h(v)\le\lambda\quad(0\le v\le1).
\end{align}
Thus all hypotheses of Lemma~\ref{lem:envelope} hold. Combining
Lemmas~\ref{lem:curvature} and~\ref{lem:envelope} gives
\begin{align}
I_{f,g}(r)
\le h'(r)^2\Phi_\lambda(d)
\le M_n(r).
\end{align}


Now suppose that $\lambda=1$. In this case, the
curvature factor simplifies to
\begin{align}\label{balancedcancel}
\Phi_1(d)
=
\frac{1}{1-d^2}.
\end{align}
Combining Lemma~\ref{lem:curvature} with the spectral
inequality \eqref{signedCS} and the identity $d=h(r)$
gives
\begin{align}
I_{f,g}(r)
&\le h'(r)^2\Phi_1(d)\notag\\
&=\frac{h'(r)^2}{1-h(r)^2}\notag\\
&\le M_n(r).
\end{align}
Thus this case follows directly from the curvature
and spectral estimates, without invoking the scalar
envelope lemma. Together, the two cases establish the
upper bound for every Boolean pair and every
$r\in(-1,1)$. Sharpness of the upper bound was verified in
Section~\ref{sec:Fisherbenchmark}.


We close this section with two remarks.
\begin{enumerate}
    \item 
Matching parities are
balanced, so $p=q=1/2$ and $a=b=1$, giving
$\lambda=\eta=1$. However, the case $\lambda=1$ also
includes biased pairs. For example, if $p=q<1/2$ and
$s(r)>0$, then
$\lambda=b=\sqrt{\omega_p/\omega_q}=1$.
Thus $\lambda=1$ does not imply that the outputs are
balanced.

\item The theorem establishes the optimality of matching
parities without asserting uniqueness. Indeed, as
observed in \cite{chen}, a broader class of matching
Boolean functions can attain the bound. Specifically,
suppose $f=g$ is a level-$k$ Boolean function, meaning
that $\widehat f(S)=0$ whenever $|S|\ne k$, where
$k\ge1$. Such a function is balanced, and Parseval's
identity gives
\begin{align}
\E_r[f(X^n)f(Y^n)]
=r^k\sum_{|S|=k}\widehat f(S)^2
=r^k.
\end{align}
Its output pair therefore has the same distribution
and Fisher information as matching degree-$k$ parities.
Consequently, it attains the upper bound whenever $k$
maximizes the parity benchmark at the correlation
under consideration.
\end{enumerate}


\section{Proof of Lemma~\ref{lem:curvature}}
\label{sec:curvature}

Fix $r\in(0,1)$ with $s(r)\ne0$, and retain the
normalization of Section~\ref{sec:normalization}.
In particular, $h$ is obtained from $s$ by the fixed
sign choice in \eqref{eq:orientation}, and
\begin{align}
d=h(r),\qquad 0<d<\lambda\le1,\qquad 0<\eta\le1.
\end{align}
We first express the four output probabilities in terms
of $d$ and the two endpoints $\lambda,\eta$. We then
compute the factor multiplying $h'(r)^2$ in the Fisher
information and bound it by $\Phi_\lambda(d)$.

\paragraph{Expressing the four output probabilities}
The definitions of $a$, $b$, and $\beta$ give
\begin{align}\label{eq:marginal_products}
pq=\beta a,\qquad
p(1-q)=\beta b,\qquad
(1-p)q=\beta b^{-1},\qquad
(1-p)(1-q)=\beta a^{-1}.
\end{align}
For example,
\begin{align}
\frac{(1-p)q}{\beta}
=
\sqrt{\frac{(1-p)q}{p(1-q)}}
=b^{-1},
\end{align}
and the other identities follow in the same way.
Since $z(v)=pq+\beta s(v)$, the general four-cell
formulas can therefore be written as
\begin{align}
&P^{f,g}_v(-1,-1)=\beta(a+s(v)),\\
&P^{f,g}_v(-1,1)=\beta(b-s(v)),\\
&P^{f,g}_v(1,-1)=\beta(b^{-1}-s(v)),\\
&P^{f,g}_v(1,1)=\beta(a^{-1}+s(v)).
\end{align}

If $s(r)>0$, then $h=s$, $\lambda=b$, and $\eta=a$.
Evaluating at $v=r$ gives
\begin{align}
&P^{f,g}_r(-1,-1)=\beta(\eta+d),\\
&P^{f,g}_r(-1,1)=\beta(\lambda-d),\\
&P^{f,g}_r(1,-1)=\beta(\lambda^{-1}-d),\\
&P^{f,g}_r(1,1)=\beta(\eta^{-1}+d).
\end{align}
If $s(r)<0$, then $h=-s$, $\lambda=a$, and $\eta=b$,
so $s(r)=-d$ and
\begin{align}
&P^{f,g}_r(-1,-1)=\beta(\lambda-d),\\
&P^{f,g}_r(-1,1)=\beta(\eta+d),\\
&P^{f,g}_r(1,-1)=\beta(\eta^{-1}+d),\\
&P^{f,g}_r(1,1)=\beta(\lambda^{-1}-d).
\end{align}
Thus, in either case, the four output probabilities,
up to their order, are
\begin{align}\label{eq:oriented_cells}
\beta(\lambda-d),\quad
\beta(\lambda^{-1}-d),\quad
\beta(\eta+d),\quad
\beta(\eta^{-1}+d).
\end{align}
All four are positive. In particular,
$\lambda-d>0$, and
$\lambda^{-1}-d\ge\lambda-d>0$ because $\lambda\le1$.

Summing \eqref{eq:oriented_cells} and using the fact
that the probabilities sum to one gives
\begin{align}
1=\beta\bigl(\lambda+\lambda^{-1}
                  +\eta+\eta^{-1}\bigr).
\end{align}
Since
\begin{align}
\lambda+\lambda^{-1}+\eta+\eta^{-1}
=
\frac{(\lambda+\eta)(1+\lambda\eta)}{\lambda\eta},
\end{align}
we obtain
\begin{align}\label{eq:beta_endpoints}
\beta=\frac{\lambda\eta}
{(\lambda+\eta)(1+\lambda\eta)}.
\end{align}

\paragraph{The exact curvature factor}
The quantities $\beta$, $\lambda$, and $\eta$ remain
fixed as the variable correlation $v$ varies: they
depend on the output marginals and the sign choice
made at the evaluation point $r$. The sign choice
also ensures that $s'(r)=h'(r)$ or $s'(r)=-h'(r)$.
Consequently, each of the four output probability
derivatives has magnitude $\beta|h'(r)|$.

Substituting these derivatives and
\eqref{eq:oriented_cells} into \eqref{eq:Fisher} yields
\begin{align}\label{eq:exact_curvature}
I_{f,g}(r)=h'(r)^2W_{\lambda,\eta}(d),
\end{align}
where
\begin{align}\label{eq:curvature_W}
W_{\lambda,\eta}(d)
=
\beta\left[
\frac{1}{\lambda-d}
+\frac{1}{\lambda^{-1}-d}
+\frac{1}{\eta+d}
+\frac{1}{\eta^{-1}+d}
\right].
\end{align}
The single factor $\beta$ in this expression results
from dividing each squared derivative, which contains
$\beta^2$, by a probability containing $\beta$.

The exact curvature factor depends on both endpoints.
We will show that it is bounded by
\begin{align}\label{auxPhi}
\Phi_\lambda(d)
=
\frac{\lambda(1-d)}
{(1+d)(\lambda-d)(1-\lambda d)},
\end{align}
which depends only on the upper endpoint $\lambda$
and the current value $d$.

\paragraph{Factoring the curvature}
Pair the first two fractions in \eqref{eq:curvature_W}
and then the last two. Using
\begin{align}
\frac{1}{\lambda^{-1}-d}
=\frac{\lambda}{1-\lambda d},
\qquad
\frac{1}{\eta^{-1}+d}
=\frac{\eta}{1+\eta d},
\end{align}
we obtain
\begin{align}
\frac{1}{\lambda-d}
+\frac{1}{\lambda^{-1}-d}
&=
\frac{1+\lambda^2-2\lambda d}
     {(\lambda-d)(1-\lambda d)},\\
\frac{1}{\eta+d}
+\frac{1}{\eta^{-1}+d}
&=
\frac{1+\eta^2+2\eta d}
     {(\eta+d)(1+\eta d)}.
\end{align}
For the following calculation, abbreviate
\begin{align}
D_-=(\lambda-d)(1-\lambda d),\qquad
D_+=(\eta+d)(1+\eta d).
\end{align}
Both denominators are positive. Bringing the paired
fractions to a common denominator gives
\begin{align}
W_{\lambda,\eta}(d)
=
\beta\,
\frac{(1+\lambda^2-2\lambda d)D_+
      +(1+\eta^2+2\eta d)D_-}
     {D_-D_+}.
\end{align}

To simplify the numerator, first write
\begin{align}
D_-=\lambda-(1+\lambda^2)d+\lambda d^2,
\qquad
D_+=\eta+(1+\eta^2)d+\eta d^2.
\end{align}
Expanding the two products and collecting powers of $d$
then gives
\begin{align}
&(1+\lambda^2-2\lambda d)D_+
 +(1+\eta^2+2\eta d)D_-\notag\\
&\quad=
\bigl[\eta(1+\lambda^2)+\lambda(1+\eta^2)\bigr]\notag\\
&\qquad\quad
-\bigl[\eta(1+\lambda^2)+\lambda(1+\eta^2)\bigr]d^2
\notag\\
&\quad=
(\lambda+\eta)(1+\lambda\eta)(1-d^2).
\label{eq:curvature_numerator}
\end{align}
Here the coefficients of $d$ and $d^3$ cancel.
Substituting \eqref{eq:beta_endpoints} therefore yields
\begin{align}\label{eq:curvature_product}
W_{\lambda,\eta}(d)
=
\frac{\lambda\eta(1-d^2)}
{(\lambda-d)(1-\lambda d)(\eta+d)(1+\eta d)}.
\end{align}
Factoring $1-d^2=(1-d)(1+d)$ and comparing with
\eqref{auxPhi}, we obtain the exact identity
\begin{align}\label{factorization}
W_{\lambda,\eta}(d)
=
\Phi_\lambda(d)\,
\frac{\eta(1+d)^2}{(\eta+d)(1+\eta d)}.
\end{align}

\paragraph{Eliminating the second endpoint}
The second factor in \eqref{factorization} is at most
one. Indeed,
\begin{align}
(\eta+d)(1+\eta d)-\eta(1+d)^2
&=d+\eta^2d-2\eta d\notag\\
&=d(1-\eta)^2
\ge0.
\end{align}
Since $(\eta+d)(1+\eta d)>0$, it follows that
\begin{align}
0<
\frac{\eta(1+d)^2}{(\eta+d)(1+\eta d)}
\le1.
\end{align}
Consequently,
\begin{align}
W_{\lambda,\eta}(d)\le\Phi_\lambda(d).
\end{align}
This comparison is valid for every permissible second
endpoint $\eta$, so it requires no equality between
the output biases. Since $d>0$, equality between the
two curvature factors holds precisely when $\eta=1$.

Finally, multiplying this bound by $h'(r)^2$ and using
\eqref{eq:exact_curvature} gives
\begin{align}\label{curvaturebound}
I_{f,g}(r)\le h'(r)^2\Phi_\lambda(d),
\end{align}
which proves Lemma~\ref{lem:curvature}.

\section{Proof of Lemma~\ref{lem:envelope}}\label{sec:scalar}

\subsection{Organization, coverage, and elementary estimates}
\label{sec:coverage}


Throughout this section, the hypotheses of
Lemma~\ref{lem:envelope} are in force. Thus
\begin{align}
h(v)=\sum_{k=1}^n a_kv^k,\qquad
\sum_{k=1}^n|a_k|\le1,\qquad
h(v)\le\lambda\quad(0\le v\le1),
\end{align}
where $n\ge1$ and $0<\lambda<1$. The evaluation point
$r\in(0,1)$ is fixed, and
\begin{align}
d=h(r)\in(0,\lambda).
\end{align}
Set
\begin{align}\label{eq:scalar_bounds_recalled}
M=M_n(r),\qquad
u=h'(r),\qquad
\delta=\lambda-d,\qquad
t=-\log r,\qquad
 \mathcal I=u^2\Phi_\lambda(d).
\end{align}
Our task is to prove $\mathcal I\le M$. If $u=0$, the definition gives $\mathcal I=0$.
We therefore concentrate on nonzero slopes.

\subsubsection{Proof strategy and coverage}

The main issue is to control $\mathcal I$ when $d$ is close
to the upper bound $\lambda$. In that situation,
the factor $\Phi_\lambda(d)$ can be large, but
the global constraint $h(v)\le\lambda$ restricts
the possible slope at $r$. We exploit this restriction
by examining the polynomial in the direction in
which its value initially increases. If $u<0$,
that direction is toward smaller values of $v$;
if $u>0$, it is toward larger values of $v$.
This explains the initial division according to
the sign of the slope.

For positive slopes, the estimates depend further
on the evaluation point and on the quantity
$nt=n(-\log r)$. Indeed, $r^k=\exp(-kt)$, so $nt$
is the largest logarithmic scale among the monomials
of degree at most $n$. The subsequent subsections
cover the following regimes:
\begin{enumerate}
\item
Section~\ref{sec:decreasing} treats $u<0$ for all
$0<r<1$ and $0<\lambda<1$.

\item
Section~\ref{sec:small} treats $u>0$ and
$0<r\le1/2$ for all $0<\lambda<1$.

\item
Section~\ref{sec:short} treats $u>0$ and $nt\le \log 2$
for all $0<\lambda<1$.

\item
Section~\ref{sec:long} treats $u>0$, $r\ge1/2$,
and $nt > \log 2$, with
$0<\lambda\le1/5$ or $4/5\le\lambda<1$.

\item
Section~\ref{sec:central} treats $u>0$, $r\ge1/2$,
and $nt > \log 2$, with $1/5\le\lambda\le4/5$.
\end{enumerate}

These regimes exhaust the hypotheses of the lemma.
After handling nonpositive slopes, either $r\le1/2$,
which is covered by Section~\ref{sec:small}, or
$r>1/2$. In the latter case, $nt\le \log 2$ is covered
by Section~\ref{sec:short}, while $nt > \log 2$ is
covered by Sections~\ref{sec:long} and
\ref{sec:central}, whose intervals for $\lambda$
together cover $(0,1)$. 

Before turning to these cases, we establish
coefficient estimates and spectral bounds that
will be used throughout the proof. In particular,
the spectral bounds below settle a substantial
part of the positive-slope domain. Subsequent
assumptions such as $d>c_0\lambda$ simply exclude
points already covered by those bounds.

Appendix~\ref{sec:forward} supplies a differential
estimate used in the two long-scale regimes.
It is valid for all positive slopes and is an
auxiliary estimate, rather than an additional
case in the partition above.

\subsubsection{Decomposing the coefficient budget}

Separate the positive and negative coefficients
of $h$ by initially setting
\begin{align}
\alpha_k=\max\{-a_k,0\},\qquad
\beta_k=\max\{a_k,0\}.
\end{align}
Then
\begin{align}
A(v)=\sum_{k=1}^n\alpha_kv^k,\qquad
B(v)=\sum_{k=1}^n\beta_kv^k
\end{align}
have nonnegative coefficients and satisfy $h=B-A$.
Their total coefficient mass is initially
$\sum_k(\alpha_k+\beta_k)=\sum_k|a_k|\le1$.

If this sum is less than one, add
\begin{align}
\frac{1-\sum_{k=1}^n|a_k|}{2}
\end{align}
to both degree-one coefficients $\alpha_1$ and
$\beta_1$, and retain the notation $A$ and $B$
for the resulting polynomials. The added terms
cancel in $B-A$, so this operation changes
neither $h$ nor any of its derivatives. It also
preserves nonnegativity of the coefficients and
the degree bound. We may therefore assume
\begin{align}
e:=A(1),\qquad o:=B(1),\qquad e+o=1.
\end{align}
The letters $e$ and $o$ denote the total coefficient
masses of $A$ and $B$, respectively; they do not
refer to even and odd degrees. After the possible
addition of equal degree-one terms, $A$ and $B$
need not be the disjoint negative and positive
parts, but they still provide a nonnegative-coefficient
decomposition $h=B-A$ with total mass exactly one. Throughout the remainder of the paper, $e=A(1)$ denotes the coefficient mass, whereas $\mathrm e=\exp(1)$ denotes Euler’s number.

Define the gap from the global upper bound by
\begin{align}\label{floor}
F(v):=\lambda-h(v)
=\lambda+A(v)-B(v)\ge0
\qquad(0\le v\le1).
\end{align}
At the evaluation point,
\begin{align}
F(r)=\delta,\qquad F'(r)=-u.
\end{align}
Thus the range constraint is equivalent to
nonnegativity of $F$ on the entire interval.
The later Taylor estimates use this nonnegativity
to relate the local gap $\delta$ to the slope $u$.

Evaluating the range constraint at $v=1$ gives
\begin{align}
h(1)=o-e=1-2e\le\lambda,
\end{align}
and therefore
\begin{align}
e\ge\frac{1-\lambda}{2},\qquad
o\le\frac{1+\lambda}{2}.
\end{align}
Also, nonnegativity of the coefficients implies
\begin{align}
d=B(r)-A(r)\le B(r)\le B(1)=o.
\end{align}
Combining these observations with $e+o=1$ yields
\begin{align}\label{massbounds}
\frac{1-\lambda}{2}\le e\le1-d,\qquad
d\le o\le\frac{1+\lambda}{2}.
\end{align}
In particular, $e>0$ because $\lambda<1$, and
$o>0$ because $d>0$.

\subsubsection{Normalizing the curvature bound}

It is useful to separate the factor $\lambda-d$
from the denominator of $\Phi_\lambda(d)$.
Define
\begin{align}
Q(\lambda,d)
=
\frac{\lambda(1-d)}
{(1+d)(1-\lambda d)}.
\end{align}
Then
\begin{align}
\Phi_\lambda(d)=\frac{Q(\lambda,d)}{\delta},
\end{align}
and the definition of $\mathcal I$ in
\eqref{eq:scalar_bounds_recalled} becomes
\begin{align}\label{eq:normalized_scalar_bound}
\frac{\mathcal I}{M}
=
\frac{u^2}{M}\frac{Q(\lambda,d)}{\delta}.
\end{align}
This form displays the quantity that must be
controlled: the squared slope relative to the
gap from the upper bound.

On $0<d<\lambda<1$, we have
$0<Q(\lambda,d)\le\lambda$. Indeed,
\begin{align}
(1+d)(1-\lambda d)-(1-d)
=d(2-\lambda-\lambda d)\ge0.
\end{align}
The function $Q$ increases with $\lambda$ and
decreases with $d$. These monotonicity properties
follow from
\begin{align}
\frac{\partial Q}{\partial\lambda}
=
\frac{1-d}{(1+d)(1-\lambda d)^2}>0
\end{align}
and
\begin{align}
\frac{\partial Q}{\partial d}
=
\frac{\lambda\bigl[-2+\lambda(1+2d-d^2)\bigr]}
{(1+d)^2(1-\lambda d)^2}<0.
\end{align}
For the last sign, note that
$1+2d-d^2=2-(1-d)^2<2$ and $\lambda<1$.

For the positive-slope arguments, we will also use
the logarithmic normalization
\begin{align}\label{lognormalization}
t=-\log r,\qquad
U=rtu,\qquad
H=Mr^2t^2
=
\max_{1\le k\le n}
\frac{(kt)^2}{\exp(2kt)-1}.
\end{align}
The final expression follows by multiplying each
term in the definition of $M_n(r)$ by $r^2t^2$
and using $r=\exp(-t)$. The factor $rt$ has a
corresponding interpretation for the slope:
along the path $v=r^{1-\theta}$,
\begin{align}
\left.\frac{d}{d\theta}h(r^{1-\theta})
\right|_{\theta=0}
=rt\,h'(r)=U.
\end{align}
Thus $U$ is the initial slope along this path, and
$H$ is the parity benchmark under the same scaling.
In particular,
\begin{align}
\frac{U^2}{H}=\frac{u^2}{M}.
\end{align}
Appendix~\ref{sec:forward} proves
\eqref{Rlemma}, \eqref{Ucaps}, and
\eqref{Zconditions}, which relate this scaled
slope to the global constraint $F\ge0$.

For reference, the principal notation is collected
below.
\begin{center}
\renewcommand{\arraystretch}{1.15}
\begin{tabular}{@{}p{0.23\textwidth}p{0.69\textwidth}@{}}
\toprule
Symbol & Meaning\\
\midrule
$h$ &
Polynomial satisfying the hypotheses of
Lemma~\ref{lem:envelope}.\\
$r$, $n$, $M$ &
Evaluation point, degree bound, and parity
benchmark $M_n(r)$.\\
$d$, $\lambda$, $\delta$ &
Value $h(r)$, global upper bound, and gap
$\lambda-d$.\\
$A,B$; $e,o$ &
Nonnegative-coefficient decomposition $h=B-A$
and masses $e=A(1)$, $o=B(1)$, with $e+o=1$.\\
$F$ &
Gap function $\lambda-h$, nonnegative on $[0,1]$.\\
$u$ &
Derivative $h'(r)$.\\
$t$ &
Logarithmic correlation scale $-\log r$.\\
$U,H$ &
Scaled slope $rtu$ and scaled benchmark $Mr^2t^2$.\\
$\theta$ &
Parameter used later for the paths
$v=r+\theta/2$ and $v=r^{1-\theta}$.\\
\bottomrule
\end{tabular}
\end{center}

\subsubsection{A spectral estimate for nonnegative slopes}

Suppose $u\ge0$. Since $A$ has nonnegative
coefficients, $A'(r)\ge0$, and hence
\begin{align}
0\le u=B'(r)-A'(r)\le B'(r).
\end{align}
The degree-by-degree estimate
\eqref{eq:degree_derivative_bound} gives
\begin{align}
B'(r)
\le
\sqrt M\sum_{k=1}^n
\beta_k\sqrt{(1-r^k)(1+r^k)}.
\end{align}
Applying Cauchy--Schwarz to the factors
$\sqrt{\beta_k(1-r^k)}$ and
$\sqrt{\beta_k(1+r^k)}$, we obtain
\begin{align}
B'(r)^2
&\le
M\left(\sum_{k=1}^n\beta_k(1-r^k)\right)
 \left(\sum_{k=1}^n\beta_k(1+r^k)\right)\notag\\
&=
M\bigl(o-B(r)\bigr)\bigl(o+B(r)\bigr)\notag\\
&=
M\bigl(o^2-B(r)^2\bigr).
\end{align}
Since $B(r)=A(r)+d\ge d>0$, it follows that
\begin{align}\label{oddcs}
u\le B'(r),\qquad
u^2\le M\bigl(o^2-B(r)^2\bigr)
\le M(o^2-d^2).
\end{align}
This estimate uses the positive-coefficient mass
$o$, which is at most $(1+\lambda)/2$ by
\eqref{massbounds}.

\subsubsection{Regions settled by the spectral estimate}

Continue to assume $u\ge0$. 
Combining the defining identity for $\mathcal I$ in 
\eqref{eq:scalar_bounds_recalled} with
\eqref{oddcs} and \eqref{massbounds} gives
\begin{align}
\frac{\mathcal I}{M}
\le
(o^2-d^2)\Phi_\lambda(d)
\le
\left(\frac{(1+\lambda)^2}{4}-d^2\right)
\Phi_\lambda(d).
\end{align}
It is therefore sufficient to verify
\begin{align}\label{eq:spectral_sufficient_test}
\left(\frac{(1+\lambda)^2}{4}-d^2\right)
\Phi_\lambda(d)\le1.
\end{align}
This test uses the coefficient masses and the
value at $r$, together with the endpoint constraint
$h(1)\le\lambda$. It does not yet use the full
constraint $F(v)\ge0$ between the endpoints.

For $0<\lambda<1$, define
\begin{align}
\mathcal Q_\lambda(d)
=
-4(1-\lambda)d^2
+\bigl[\lambda(1-\lambda)-4\bigr]d
+\lambda^2+3\lambda.
\end{align}
Substitution of the formula for $\Phi_\lambda$
and collection of terms gives the exact identity
\begin{align}\label{qpoly}
1-
\left(\frac{(1+\lambda)^2}{4}-d^2\right)
\Phi_\lambda(d)
=
\frac{(1-\lambda)\mathcal Q_\lambda(d)}
{4(1+d)(\lambda-d)(1-\lambda d)}.
\end{align}
The denominator and the factor $1-\lambda$ are
positive. Thus \eqref{eq:spectral_sufficient_test}
holds if and only if $\mathcal Q_\lambda(d)\ge0$.

Moreover,
\begin{align}
\mathcal Q_\lambda'(d)
=
-8(1-\lambda)d+\lambda(1-\lambda)-4
\le-\frac{15}{4}<0
\qquad(d\ge0).
\end{align}
Consequently, if
$\mathcal Q_\lambda(\bar d)\ge0$ at a proposed
upper bound $\bar d$, then the spectral test holds
for every admissible $d\le\bar d$.
We use this observation to obtain several explicit
cutoffs suited to the later arguments.

First, the desired bound $\mathcal I\le M$ holds
whenever any of the following conditions is satisfied:
\begin{align}
d&\le c_0\lambda,
&&c_0=\frac{741}{1000};
\label{c0}\\
d&\le2\lambda-1,
&&\lambda\ge\frac12;
\label{twolambda}\\
d&\le\frac{3\lambda}{4},
&&\lambda\ge\frac13.
\label{threequarters}
\end{align}
To verify \eqref{c0}, substitute $d=c_0\lambda$:
\begin{align}
\frac{\mathcal Q_\lambda(c_0\lambda)}{\lambda}
={}&
(4c_0^2-c_0)\lambda^2
+(1-4c_0^2+c_0)\lambda+3-4c_0\notag\\
={}&
\frac{363831}{250000}\lambda^2
-\frac{113831}{250000}\lambda
+\frac9{250}.
\end{align}
This quadratic has positive leading coefficient
and negative discriminant, since
\begin{align}
113831^2-36000\cdot363831
=-140419439<0.
\end{align}
It is therefore positive for every real $\lambda$,
which proves \eqref{c0}. The other two cutoffs
follow from
\begin{align}
\mathcal Q_\lambda(2\lambda-1)
=14\lambda(1-\lambda)^2
\end{align}
and
\begin{align}
\mathcal Q_\lambda(3\lambda/4)
=\frac{\lambda^2(3\lambda-1)}{2}.
\end{align}
These values are nonnegative on the respective
ranges in \eqref{twolambda} and
\eqref{threequarters}.

For $\lambda\ge1/3$, a slightly stronger cutoff
will also be useful:
\begin{align}\label{dstar}
d\le d_*(\lambda)
:=
\frac{3\lambda}{4}
+\frac{\lambda^2(3\lambda-1)}{12}.
\end{align}
To prove this, first observe that
$3\lambda/4\le d_*(\lambda)\le\lambda$.
Indeed,
\begin{align}
\frac{d_*(\lambda)}{\lambda}
=
\frac34+\frac{\lambda(3\lambda-1)}{12}
\le\frac34+\frac16
=\frac{11}{12}<1.
\end{align}
On $0\le d\le\lambda$,
\begin{align}
-\mathcal Q_\lambda'(d)
&=
8(1-\lambda)d-\lambda(1-\lambda)+4\notag\\
&\le4+7\lambda(1-\lambda)
\le\frac{23}{4}<6.
\end{align}
Thus, starting from $d=3\lambda/4$, increasing
$d$ by an amount $\Delta$ decreases
$\mathcal Q_\lambda(d)$ by at most $6\Delta$.
Taking
\begin{align}
\Delta
=
d_*(\lambda)-\frac{3\lambda}{4}
=
\frac{\lambda^2(3\lambda-1)}{12}
\end{align}
gives
\begin{align}
\mathcal Q_\lambda(d_*(\lambda))
&\ge
\mathcal Q_\lambda(3\lambda/4)-6\Delta\notag\\
&=
\frac{\lambda^2(3\lambda-1)}2
-\frac{\lambda^2(3\lambda-1)}2
=0.
\end{align}
Monotonicity of $\mathcal Q_\lambda$ now proves
the cutoff \eqref{dstar}.

The cutoffs \eqref{c0}--\eqref{dstar} are sufficient
conditions for the scalar inequality, chosen for
their use in the subsequent cases. They need not
be sharp. Their purpose is to settle points that
are sufficiently far from the upper bound
$\lambda$, leaving the later arguments to exploit
the full constraint $F\ge0$ near that bound.
In particular, once \eqref{c0} has been excluded,
we have
\begin{align}
d>c_0\lambda,\qquad
\delta=\lambda-d<(1-c_0)\lambda.
\end{align}
This restriction is used in the central
concavity argument.

\subsection{\texorpdfstring{$u<0$: all $0<r<1$ and
$0<\lambda<1$}{Negative slope: all correlations and ceilings}}
\label{sec:decreasing}

Assume $u=h'(r)<0$. This subsection treats every
$0<r<1$ and $0<\lambda<1$ under this slope condition;
no additional restriction on the parameters is needed.
We use only the defining identity for $\mathcal I$,
\begin{align}
\mathcal I=u^2\Phi_\lambda(d).
\end{align}


Recall that $h=B-A$ and $F=\lambda-h\ge0$ on $[0,1]$.
Since $h'(r)<0$, moving backward from $r$ initially
increases $h$, or equivalently decreases $F$.
The constraint $F\ge0$ therefore limits how large
$|h'(r)|$ can be when $F(r)=\delta$ is small.
We make this observation quantitative by combining
a bound on the slope with an upper bound on $F''$
along the backward interval.

Put
\begin{align}
V=-u=F'(r)>0,\qquad
C=1-d>0,\qquad
A_r=A(r).
\end{align}
All sums below run over $1\le k\le n$.

\paragraph{Step 1: remove the degree-one contribution}

Since $V=A'(r)-B'(r)$, we have
\begin{align}
rV
&=\sum_k k\alpha_kr^k-\sum_k k\beta_kr^k\notag\\
&=\sum_k(k-1)\alpha_kr^k
  -\sum_k(k-1)\beta_kr^k+A(r)-B(r)\notag\\
&=\sum_k(k-1)\alpha_kr^k
  -\sum_k(k-1)\beta_kr^k-d\notag\\
&\le\sum_k(k-1)\alpha_kr^k.
\label{eq:negative_degree_removal}
\end{align}
The last step uses $\beta_k\ge0$ and $d>0$.
The factor $k-1$ eliminates the degree-one term.
This is useful because the remaining terms can
be controlled by the same parity estimates that
bound the second derivative of $A$.

We claim that, for every $1\le k\le n$,
\begin{align}\label{degreebounds}
(k-1)r^{k-1}
\le2\sqrt M\sqrt{r^k(1-r^k)},\qquad
k(k-1)r^{k-2}
\le4M(1-r^k).
\end{align}
For $k=1$, both left sides are zero.
For $k\ge2$, it suffices to prove the second
inequality: multiplying it by $r^k$ and using
$(k-1)^2\le k(k-1)$ gives
\begin{align}
(k-1)^2r^{2k-2}
\le4Mr^k(1-r^k),
\end{align}
whose square root is the first inequality.

To prove the second inequality, first suppose
$k=2j$. The degree-$j$ term in the definition
of $M=M_n(r)$ gives
\begin{align}
M\ge
\frac{j^2r^{2j-2}}{1-r^{2j}}
=
\frac{j^2r^{k-2}}{1-r^k}.
\end{align}
Since $k(k-1)=2j(2j-1)\le4j^2$, it follows that
\begin{align}
k(k-1)r^{k-2}\le4M(1-r^k).
\end{align}

Next suppose $k=2j+1\ge3$. Both $j$ and $j+1$
are at most $n$. Since $M$ dominates the parity
terms of these two degrees, it also dominates
their geometric mean:
\begin{align}
M
&\ge
\left(
\frac{j^2r^{2j-2}}{1-r^{2j}}
\frac{(j+1)^2r^{2j}}{1-r^{2j+2}}
\right)^{1/2}\notag\\
&=
\frac{j(j+1)r^{2j-1}}
{\sqrt{(1-r^{2j})(1-r^{2j+2})}}.
\end{align}
Moreover,
\begin{align}
&(1-r^{2j+1})^2
-(1-r^{2j})(1-r^{2j+2})\notag\\
&=r^{2j}(1-r)^2\ge0,
\end{align}
so
\begin{align}
M\ge
\frac{j(j+1)r^{k-2}}{1-r^k}.
\end{align}
Finally,
$k(k-1)=2j(2j+1)\le4j(j+1)$, which proves
the second inequality in \eqref{degreebounds}
also for odd $k$.

\paragraph{Step 2: bound the slope and the second derivative}

Dividing \eqref{eq:negative_degree_removal} by
$r>0$ and applying the first inequality in
\eqref{degreebounds}, we obtain
\begin{align}
V
\le
2\sqrt M\sum_k\alpha_k\sqrt{r^k(1-r^k)}.
\end{align}
Cauchy--Schwarz gives
\begin{align}
\left(\sum_k\alpha_k\sqrt{r^k(1-r^k)}\right)^2
&\le
\left(\sum_k\alpha_kr^k\right)
\left(\sum_k\alpha_k(1-r^k)\right)\notag\\
&=A_r(e-A_r).
\end{align}
Consequently,
\begin{align}\label{eq:negative_slope_mass}
V^2\le4MA_r(e-A_r).
\end{align}

The second inequality in \eqref{degreebounds}
similarly yields
\begin{align}
A''(r)
&=\sum_k k(k-1)\alpha_kr^{k-2}\notag\\
&\le4M\sum_k\alpha_k(1-r^k)\notag\\
&=4M(e-A_r).
\end{align}
Because $A$ and $B$ have nonnegative coefficients,
$B''(v)\ge0$ and $A''(v)$ is nondecreasing on
$[0,1]$. Therefore, for $0\le v\le r$,
\begin{align}\label{eq:negative_second_derivative}
F''(v)
=A''(v)-B''(v)
\le A''(v)
\le A''(r)
\le4M(e-A_r)=:K_0.
\end{align}
Thus the same quantity $e-A_r$ controls both
the slope estimate and an upper bound on the
second derivative along the backward interval.

\paragraph{Step 3: use nonnegativity along the backward interval}

We first verify that the Taylor step used below
lies inside $[0,r]$. Since
\begin{align}
F(0)=\lambda>\lambda-d=F(r),
\end{align}
the mean value theorem gives a point
$v_0\in(0,r)$ such that
\begin{align}
F'(v_0)=\frac{F(r)-F(0)}{r}
=-\frac dr<0.
\end{align}
Using \eqref{eq:negative_second_derivative},
\begin{align}
V-F'(v_0)
=
\int_{v_0}^r F''(v)\,dv
\le K_0(r-v_0).
\end{align}
The left side is positive, so $K_0>0$.
Since $F'(v_0)<0$, the same inequality implies
\begin{align}
V<K_0(r-v_0)<K_0r.
\end{align}
Hence
\begin{align}
\ell:=\frac{V}{K_0}
\end{align}
satisfies $0<\ell<r$.

Taylor's formula with integral remainder gives
\begin{align}
F(r-\ell)
&=
F(r)-\ell F'(r)
+\int_{r-\ell}^r
\bigl(v-(r-\ell)\bigr)F''(v)\,dv\notag\\
&\le
\delta-\ell V+\frac{K_0\ell^2}{2}.
\end{align}
The integral weight is nonnegative, so the upper
bound on $F''$ applies directly. Since
$F(r-\ell)\ge0$ and $\ell=V/K_0$, we conclude that
\begin{align}
0\le\delta-\frac{V^2}{2K_0}.
\end{align}
Equivalently,
\begin{align}
V^2\le2K_0\delta
=8M(e-A_r)\delta.
\end{align}
Combining this with
\eqref{eq:negative_slope_mass} gives
\begin{align}\label{eq:negative_combined_mass}
V^2
\le4M(e-A_r)\min\{A_r,2\delta\}.
\end{align}
This is the required relation between the slope
and the gap: when $\delta$ is small, nonnegativity
of $F$ forces $V^2$ to be small as well.

\paragraph{Step 4: eliminate the coefficient masses}

Since
\begin{align}
B(r)=A_r+d\le B(1)=o=1-e,
\end{align}
we have
\begin{align}
e-A_r\le1-d-2A_r=C-2A_r.
\end{align}
Also $e-A_r\ge0$, so $0\le A_r\le C/2$.
It follows from \eqref{eq:negative_combined_mass}
that
\begin{align}
\frac{V^2}{M}
\le
4(C-2A_r)\min\{A_r,2\delta\}.
\end{align}
We now maximize the expression on the right
over $0\le A_r\le C/2$.

For $A_r\le2\delta$, this expression is the
concave quadratic $4A_r(C-2A_r)$, whose maximum
occurs at $A_r=C/4$. For $A_r\ge2\delta$, it is
$8\delta(C-2A_r)$, which decreases with $A_r$.
If $2\delta\le C/4$, the two pieces therefore
attain their common maximum at $A_r=2\delta$.
If $2\delta\ge C/4$, the first piece attains
the overall maximum at $A_r=C/4$. Thus
\begin{align}\label{eq:negative_slope_final}
\frac{V^2}{M}
\le
\begin{cases}
8\delta(C-4\delta),&0<\delta\le C/8,\\[2pt]
C^2/2,&\delta\ge C/8.
\end{cases}
\end{align}
The two expressions agree when $\delta=C/8$.

\paragraph{Step 5: apply the defining identity for $\mathcal I$}

Since $u^2=V^2$, the defining identity for $\mathcal I$ gives
\begin{align}
\frac{\mathcal I}{M}
=\frac{V^2}{M}\Phi_\lambda(d).
\end{align}
We show that each branch of
\eqref{eq:negative_slope_final}, multiplied by
$\Phi_\lambda(d)$, is at most one.

First suppose $\delta\ge C/8$, and set
\begin{align}
\lambda_0=d+\frac C8=\frac{1+7d}{8}.
\end{align}
Then $d<\lambda_0\le\lambda<1$.
For fixed $d$, the curvature factor is
nonincreasing in $\lambda$, because
\begin{align}
\frac{\partial\Phi_\lambda(d)}{\partial\lambda}
=
-\frac{d(1-d)(1-\lambda^2)}
{(1+d)(\lambda-d)^2(1-\lambda d)^2}
\le0.
\end{align}
Consequently,
\begin{align}
\frac{\mathcal I}{M}
&\le\frac{C^2}{2}\Phi_\lambda(d)\notag\\
&\le\frac{C^2}{2}\Phi_{\lambda_0}(d)\notag\\
&=
\frac{4\lambda_0C^2}
{(1+d)(1-\lambda_0d)}.
\end{align}
The last equality uses $1-d=C$ and
$\lambda_0-d=C/8$. To show that this ratio is
at most one, subtract its numerator from its
positive denominator:
\begin{align}
&(1+d)(1-\lambda_0d)-4\lambda_0C^2\notag\\
&=
\frac{C}{8}(35d^2-9d+4).
\end{align}
The right side is positive, since $C>0$ and
\begin{align}
35d^2-9d+4
=
35\left(d-\frac9{70}\right)^2
+\frac{479}{140}>0.
\end{align}
This proves the desired bound when $\delta\ge C/8$.

Now suppose $0<\delta\le C/8$. Using the first
branch of \eqref{eq:negative_slope_final} and
cancelling $\delta=\lambda-d$, we obtain
\begin{align}
\frac{\mathcal I}{M}
&\le8\delta(C-4\delta)\Phi_\lambda(d)\notag\\
&=
\frac{8\lambda C(C-4\delta)}
{(1+d)(1-\lambda d)}.
\end{align}
It is therefore enough to show
\begin{align}
N:=(1+d)(1-\lambda d)
-8\lambda C(C-4\delta)\ge0.
\end{align}
Introduce the dimensionless variable
\begin{align}
w=\frac{8\delta}{C}\in(0,1].
\end{align}
Then
\begin{align}
\delta=\frac{w(1-d)}8,\qquad
\lambda=d+\frac{w(1-d)}8.
\end{align}
Substituting these identities into $N$ and
collecting powers of $d$ gives
\begin{align}
\frac{8N}{C}
={}&(72-41w+4w^2)d^2\notag\\
&+(-48+47w-8w^2)d
 +(8-8w+4w^2).
\label{eq:negative_gap_quadratic}
\end{align}
For $0\le w\le1$, the leading coefficient is
positive; for example,
\begin{align}
72-41w+4w^2\ge31>0.
\end{align}
The discriminant of the quadratic in $d$ is
\begin{align}
&(-48+47w-8w^2)^2\notag\\
&\quad
-4(72-41w+4w^2)(8-8w+4w^2)\notag\\
&=w(32w^2+385w-896)\le0.
\end{align}
Indeed,
\begin{align}
32w^2+385w-896
\le32+385-896=-479<0
\end{align}
on this interval. A quadratic with positive
leading coefficient and nonpositive discriminant
is nonnegative for every real value of its
argument. Thus \eqref{eq:negative_gap_quadratic}
implies $N\ge0$, proving the desired bound also
when $\delta\le C/8$.

We have therefore established $\mathcal I\le M_n(r)$
for every negative slope $u<0$, throughout the full
ranges $0<r<1$ and $0<\lambda<1$.

\subsection{\texorpdfstring{$u>0$, $0<r\le1/2$:
all $0<\lambda<1$}{Positive slope: small correlations,
all ceilings}}
\label{sec:small}

Assume $u=h'(r)>0$ and $0<r\le1/2$. We prove the
scalar inequality throughout $0<\lambda<1$ in this
regime. Since the slope is positive, moving forward
from $r$ initially increases $h$ and decreases the
gap function $F=\lambda-h$. We use $F\ge0$ along
a forward path to bound $u^2/\delta$, where
$\delta=F(r)=\lambda-d$.

The restriction $r\le1/2$ also gives
$r^k\le1/2$ for every positive degree $k$.
Consequently, the value of $B(r)$ is at most
half its total coefficient mass. This additional
constraint will allow us to complete the proof:
for $\lambda\ge1/2$, the earlier spectral test
suffices, while for $\lambda<1/2$ we combine
the mass constraint with the forward estimate.

\subsubsection{An elementary logarithmic inequality}

We will use
\begin{align}\label{logineq}
x-\log(1+x)
\ge
\frac{x^2}{2+(4/3)x},
\qquad x\ge0.
\end{align}
To verify it, subtract the right side from the
left side. The difference vanishes at $x=0$,
and its derivative is
\begin{align}
\frac{d}{dx}
\left(
x-\log(1+x)-\frac{x^2}{2+(4/3)x}
\right)
=
\frac{x^3}{(1+x)(3+2x)^2}\ge0.
\end{align}
Thus the difference is nonnegative on $[0,\infty)$.

\subsubsection{A bound on the second derivative along a forward path}

We first establish the monomial estimate
\begin{align}\label{eq:small_monomial_curvature}
k(k-1)v^{k-2}
\le\frac{3}{4(1-v)^2},
\qquad \frac12\le v<1,\quad k\ge2.
\end{align}
Equivalently, we must bound
$k(k-1)v^{k-2}(1-v)^2$ by $3/4$.

For $k=2$, this expression is
$2(1-v)^2\le1/2$. For $k=3$, it is
$6v(1-v)^2$, whose derivative is
$6(1-v)(1-3v)\le0$ on $[1/2,1)$.
Its maximum on this interval is therefore its
value $3/4$ at $v=1/2$.

For $k\ge4$, differentiation of
$v^{k-2}(1-v)^2$ shows that its maximum on
$[1/2,1]$ occurs at
\begin{align}
v=\frac{k-2}{k}=1-\frac2k.
\end{align}
The corresponding maximum is
\begin{align}
4\left(1-\frac1k\right)
 \left(1-\frac2k\right)^{k-2}.
\end{align}
This quantity equals $3/4$ at $k=4$ and is
nonincreasing thereafter. Indeed, treating $k$
as a real variable, its logarithmic derivative is
\begin{align}
\frac{1}{k(k-1)}
+\log\left(1-\frac2k\right)+\frac2k.
\end{align}
The elementary inequality
$\log(1-y)\le-y-y^2/2$ for $0\le y<1$
gives
\begin{align}
\frac{1}{k(k-1)}
+\log\left(1-\frac2k\right)+\frac2k
\le
\frac{1}{k(k-1)}-\frac2{k^2}
\le0.
\end{align}
This proves \eqref{eq:small_monomial_curvature}
for all $k\ge2$.

Consider now the forward path
\begin{align}
v=r+\frac{\theta}{2},
\qquad 0\le\theta<1.
\end{align}
Because $r\le1/2$, this path remains in $[0,1)$
and satisfies
\begin{align}
r+\frac{\theta}{2}
\le\frac{1+\theta}{2}.
\end{align}
The polynomial $A''$ is nondecreasing on $[0,1]$,
since $A$ has nonnegative coefficients. Applying
\eqref{eq:small_monomial_curvature} at
$v=(1+\theta)/2$ therefore gives
\begin{align}
A''\left(r+\frac{\theta}{2}\right)
&\le A''\left(\frac{1+\theta}{2}\right)\notag\\
&\le
\sum_{k=2}^n
\alpha_k\frac{3}{(1-\theta)^2}\notag\\
&\le\frac{3e}{(1-\theta)^2}.
\label{eq:small_A_curvature}
\end{align}

Define
\begin{align}
G(\theta)=F\left(r+\frac{\theta}{2}\right).
\end{align}
Then $G(\theta)\ge0$, and
\begin{align}
G(0)=\delta,\qquad
G'(0)=\frac12F'(r)=-\frac u2.
\end{align}
Since $F''=A''-B''\le A''$, the chain rule and
\eqref{eq:small_A_curvature} yield
\begin{align}
G''(\theta)
=
\frac14F''\left(r+\frac{\theta}{2}\right)
\le\frac{3e/4}{(1-\theta)^2}.
\end{align}

Taylor's formula with integral remainder now gives
\begin{align}
G(\theta)
&=
\delta-\frac u2\theta
+\int_0^\theta(\theta-y)G''(y)\,dy\notag\\
&\le
\delta-\frac u2\theta
+\frac{3e}{4}
\int_0^\theta\frac{\theta-y}{(1-y)^2}\,dy\notag\\
&=
\delta-\frac u2\theta
+\frac{3e}{4}\bigl[-\log(1-\theta)-\theta\bigr].
\end{align}
Combining this upper bound with $G(\theta)\ge0$
shows that
\begin{align}\label{eq:small_forward_gap}
\delta
\ge
\frac u2\theta
-\frac{3e}{4}\bigl[-\log(1-\theta)-\theta\bigr]
\qquad(0\le\theta<1).
\end{align}

\subsubsection{Optimizing the forward step}

For this calculation, temporarily put
\begin{align}
a_0=\frac u2,\qquad b_0=\frac{3e}{4}.
\end{align}
Both quantities are positive: $u>0$, and
\eqref{massbounds} gives
$e\ge(1-\lambda)/2>0$.
The right side of \eqref{eq:small_forward_gap} is
\begin{align}
(a_0+b_0)\theta+b_0\log(1-\theta).
\end{align}
Its first and second derivatives are
\begin{align}
a_0+b_0-\frac{b_0}{1-\theta},
\qquad
-\frac{b_0}{(1-\theta)^2},
\end{align}
respectively. It is therefore strictly concave
and attains its maximum at
\begin{align}
\theta=\frac{a_0}{a_0+b_0}\in(0,1).
\end{align}
Substitution into \eqref{eq:small_forward_gap}
gives
\begin{align}
\delta
\ge
a_0-b_0\log\left(1+\frac{a_0}{b_0}\right).
\end{align}
Applying \eqref{logineq} with $x=a_0/b_0$
and multiplying by $b_0$ yields
\begin{align}
\delta
\ge
\frac{a_0^2}{2b_0+(4/3)a_0}
=
\frac{(u/2)^2}{3e/2+2u/3}.
\end{align}
Rearranging, we obtain
\begin{align}
\frac{u^2}{\delta}\le6e+\frac83u.
\end{align}

To bound the remaining occurrence of $u$, note that
\begin{align}
kr^{k-1}\le\frac{k}{2^{k-1}}\le1
\qquad(k\ge1),
\end{align}
where the last inequality is equivalent to
$k\le2^{k-1}$. Hence
\begin{align}
u=B'(r)-A'(r)\le B'(r)
=\sum_{k=1}^n k\beta_kr^{k-1}
\le\sum_{k=1}^n\beta_k=o.
\end{align}
We have thus proved
\begin{align}\label{smallbound}
\frac{u^2}{\delta}
\le6e+\frac83u
\le6e+\frac83o.
\end{align}

\subsubsection{The range \texorpdfstring{$1/2\le\lambda<1$}
{one half at most lambda less than one}}

Since $r^k\le1/2$ for every $k\ge1$,
\begin{align}
d\le B(r)
=\sum_{k=1}^n\beta_kr^k
\le\frac o2
\le\frac{1+\lambda}{4}.
\label{eq:small_value_bound}
\end{align}
For $\lambda\ge1/2$, the upper bound
$(1+\lambda)/4$ is strictly less than $\lambda$,
so it lies in the domain of the spectral test.

Recall from \eqref{qpoly} that this test succeeds
when $\mathcal Q_\lambda(d)\ge0$, and that
$\mathcal Q_\lambda$ decreases with $d$.
It is therefore enough to evaluate it at the
upper bound in \eqref{eq:small_value_bound}.
Direct substitution gives
\begin{align}
\mathcal Q_\lambda\left(\frac{1+\lambda}{4}\right)
=
\frac{5\lambda^2+8\lambda-5}{4}.
\end{align}
The polynomial $5\lambda^2+8\lambda-5$ is
increasing for $\lambda\ge0$, and its value at
$\lambda=1/2$ is $1/4$. Thus it is positive
throughout the present range. Equivalently,
the numerator in \eqref{qpoly}, when its
denominator is written as
$(1+d)(\lambda-d)(1-\lambda d)$, equals
\begin{align}
\frac{(1-\lambda)(5\lambda^2+8\lambda-5)}{16}
>0
\end{align}
at this upper endpoint. The spectral estimate
therefore proves $\mathcal I\le M$ for
$1/2\le\lambda<1$.

\subsubsection{The range \texorpdfstring{$0<\lambda\le1/2$}
{zero less than lambda at most one half}}

The cutoff \eqref{c0} already proves the desired
bound when $d\le c_0\lambda$, where
$c_0=741/1000$. We may therefore assume
\begin{align}
d>c_0\lambda.
\end{align}
By \eqref{eq:small_value_bound}, $o\ge2d$,
and hence $e=1-o\le1-2d$. Using $e+o=1$,
we obtain
\begin{align}
6e+\frac83o
&=\frac83+\frac{10}{3}e\notag\\
&\le\frac83+\frac{10}{3}(1-2d)\notag\\
&=6-\frac{20}{3}d.
\end{align}
The defining identity for $\mathcal I$, together
with $\Phi_\lambda(d)=Q(\lambda,d)/\delta$
and \eqref{smallbound}, now gives
\begin{align}
\mathcal I
=
Q(\lambda,d)\frac{u^2}{\delta}
\le
Q(\lambda,d)\left(6-\frac{20}{3}d\right).
\end{align}
Both factors on the right are positive and
decrease with $d$ on $0\le d\le\lambda\le1/2$.
Indeed, the monotonicity of $Q$ was established
in Section~\ref{sec:coverage}, and
$6-20d/3\ge8/3>0$ on this interval.
Consequently,
\begin{align}
\mathcal I
\le
Q(\lambda,c_0\lambda)
\left(6-\frac{20}{3}c_0\lambda\right).
\label{eq:small_ceiling_bound}
\end{align}
We will show that the last expression is at
most one.

Using the definition of $Q$, the expression
in \eqref{eq:small_ceiling_bound} is
\begin{align}
\frac{\lambda(1-c_0\lambda)}
{(1+c_0\lambda)(1-c_0\lambda^2)}
\left(6-\frac{20}{3}c_0\lambda\right).
\end{align}
Its denominator is positive for
$0\le\lambda\le1/2$. Subtracting its numerator
from its denominator gives
\begin{align}
N(\lambda)
:={}&
(1+c_0\lambda)(1-c_0\lambda^2)\notag\\
&-\lambda(1-c_0\lambda)
\left(6-\frac{20}{3}c_0\lambda\right)\notag\\
={}&
1-\frac{5259}{1000}\lambda
+\frac{1729}{200}\lambda^2
-\frac{4209621}{1000000}\lambda^3.
\label{eq:small_cubic}
\end{align}
Thus it remains to prove $N(\lambda)\ge0$
on $[0,1/2]$.

For an explicit verification, write
\begin{align}
z=\lambda-\frac9{20},
\qquad -\frac9{20}\le z\le\frac1{20}.
\end{align}
Here $z$ is only a local scalar variable in this
polynomial calculation. Expanding
\eqref{eq:small_cubic} about $\lambda=9/20$ gives
\begin{align}
N(\lambda)
={}&
-\frac{4209621}{1000000}z^3
+\frac{59240233}{20000000}z^2\notag\\
&-\frac{14337903}{400000000}z
+\frac{3686291}{8000000000}.
\end{align}
To bound the cubic and quadratic terms together,
use $z\le1/20$:
\begin{align}
&-\frac{4209621}{1000000}z^3
+\frac{59240233}{20000000}z^2\notag\\
&=
z^2\left(
\frac{59240233}{20000000}
-\frac{4209621}{1000000}z
\right)\notag\\
&\ge
\frac{55030612}{20000000}z^2
\ge\frac{11}{4}z^2.
\end{align}
Also,
\begin{align}
\frac{14337903}{400000000}<\frac9{250},
\qquad
\frac{3686291}{8000000000}>\frac9{20000}.
\end{align}
Therefore
\begin{align}
N(\lambda)
\ge
\frac{11}{4}z^2-\frac9{250}|z|
+\frac9{20000}.
\end{align}
The expression on the right is a quadratic
in $|z|$ with positive leading coefficient and
negative discriminant:
\begin{align}
\left(\frac9{250}\right)^2
-4\cdot\frac{11}{4}\cdot\frac9{20000}
=
-\frac{1827}{500000}<0.
\end{align}
It is consequently positive for every real $z$.
This proves $N(\lambda)>0$ on $[0,1/2]$ and
hence $\mathcal I\le1$ in the remaining
small-ceiling range.

Finally, the degree-one term in the parity
benchmark gives
\begin{align}
M=M_n(r)\ge\frac{1}{1-r^2}\ge1.
\end{align}
Thus $\mathcal I\le M$ also for
$0<\lambda\le1/2$. Together with the preceding
case, this completes the proof for $u>0$,
$0<r\le1/2$, and every $0<\lambda<1$.

\subsection{\texorpdfstring{$u>0$, $nt \le \log 2$:
all $0<\lambda<1$}
{Positive slope: short logarithmic scale, all ceilings}}
\label{sec:short}
The condition $nt \le \log 2$ is equivalent to  $r^n\ge1/2$.
\begin{proposition}
[A one-step Taylor estimate]\label{lem:short_taylor_new}
Under the hypotheses and notation of Section~\ref{sec:coverage},
suppose $u=h'(r)>0$ and $x=r^n\ge1/2$. Then 
\begin{align}\label{eq:short_slope_new}
 \frac{u^2}{M_n(r)}\le 2o(\lambda-d).
\end{align}
Consequently,
\begin{align}\label{eq:short_bound_new}
 \frac{\mathcal I}{M_n(r)}
 \le 2oQ(\lambda,d)\le(1+\lambda)Q(\lambda,d).
\end{align}
\end{proposition}
\begin{proof}
Write $p_k=k/n\in(0,1]$, and set
\begin{align}
 g(z)=h(z^{1/n})=\sum_{k=1}^n(\beta_k-\alpha_k)z^{p_k},
 \qquad b=g'(x)=\frac{u}{nr^{n-1}}>0.
\end{align}
For $1/2\le z\le1$ and $0<p\le1$,
\begin{align}
 pz^{p-1}\le1,\qquad p(1-p)z^{p-2}\le1.\label{eq:bounds}
\end{align}
Indeed, $p\mapsto pz^{p-1}$ is increasing because
$1+p\log z\ge1-\log2>0$; the second inequality follows
from $p(1-p)\le1/4$ and $z^{p-2}\le z^{-2}\le4$.
Since the coefficients of $A$ and $B$ are nonnegative,
\begin{align}
 &b= \sum_{k=1}^n(\beta_k-\alpha_k)p_k x^{p_k-1} \le \sum_{k=1}^n \beta_k = o, \\
 &g''(z)= \sum_{k=1}^{n-1}(\beta_k-\alpha_k)p_k (p_k-1) z^{p_k-2} \ge - \sum_{k=1}^{n-1} \beta_k \ge - o.
\end{align}
Here the degree-$n$ term vanishes because $p_n=1$.
The bounds follow from \eqref{eq:bounds}, applied at $z=x$
and throughout $x\le z\le1$, respectively.
Hence,
\begin{align}
 0<b\le o,\qquad g''(z)\ge-o\quad(x\le z\le1).
\end{align}
Choose $s=(1-x)b/o\in(0,1-x]$. Taylor's theorem
and $g\le\lambda$ imply
\begin{align}
\delta&=\lambda-g(x)\notag\\
&\ge g(x+s)-g(x)\notag\\
& \ge bs-\frac{o}{2}s^2\notag\\
&=\frac{(1-x^2)b^2}{2o}.
\end{align}
On the other hand, the degree-$n$ term in $M_n(r)$ yields
\begin{align}
\frac{u^2}{M_n(r)}
\le(1-r^{2n})\left(\frac{u}{nr^{n-1}}\right)^2
=(1-x^2)b^2.
\end{align}
Combining these two inequalities gives \eqref{eq:short_slope_new}.
Now multiply by $\Phi_\lambda(d)=Q(\lambda,d)/\delta$
and use $2o\le1+\lambda$ to obtain \eqref{eq:short_bound_new}.
\end{proof}

If $\lambda\le1/2$, \eqref{eq:short_bound_new} and
$Q(\lambda,d)\le\lambda$ give
\begin{align}
 \frac{\mathcal I}{M}\le\lambda(1+\lambda)
 \le\frac34<1.
\end{align}
For $\lambda\ge1/2$, the spectral cutoff
\eqref{twolambda} already covers $d\le2\lambda-1$.
If $d>2\lambda-1$, monotonicity of $Q(\lambda,d)$
in $d$ yields
\begin{align}
 \frac{\mathcal I}{M}
 \le(1+\lambda)Q(\lambda,d)
 \le(1+\lambda)Q(\lambda,2\lambda-1)
 =\frac{1+\lambda}{1+2\lambda}\le\frac34.
\end{align}
Therefore  $\mathcal I\le M$  holds whenever
$u>0$ and $nt \le \log 2$.

\subsection{\texorpdfstring{$u>0$, $r\ge 1/2$, $nt > \log 2$:
$0<\lambda\le1/5$ or $4/5\le\lambda<1$}
{Positive slope: long logarithmic scale, outer ceilings}}
\label{sec:long}

Assume $u>0$, $r\ge1/2$, and $nt > \log 2$, where
$t=-\log r$. We prove the desired bound in the two
outer ranges
\begin{align}
0<\lambda\le\frac15
\qquad\text{and}\qquad
\frac45\le\lambda<1.
\end{align}
The remaining range of $\lambda$ is treated in
Section~\ref{sec:central}.

We use the differential estimate from
Appendix~\ref{sec:forward}. Recall that
\begin{align}
U=rtu,\qquad
H=Mr^2t^2
=\max_{1\le k\le n}
\frac{(kt)^2}{\exp(2kt)-1}.
\end{align}
Whenever $mH\ge1$, define
\begin{align}
Z=\frac{U}{\sqrt{mH}},\qquad
A_*=\frac{13e}{12}+\frac{4Z}{3}.
\end{align}
Then \eqref{Zconditions} gives
\begin{align}
d+Z\le o,\qquad d^2+mZ^2\le o^2,
\end{align}
and the estimate \eqref{coarse} gives
\begin{align}\label{eq:outer_coarse}
\frac{\mathcal I}{M}\le mQ(\lambda,d)A_*.
\end{align}
Our approach is to choose $m$ independently of $n$
and $r$, use the two constraints on $d$ and $Z$
to bound $A_*$ in terms of $d$, and then apply
the spectral cutoffs from Section~\ref{sec:coverage}.

\subsubsection{Lower bounds on the normalized parity envelope}

Define
\begin{align}
\psi(x)=\frac{x^2}{\exp(2x)-1},\qquad x>0.
\end{align}
We first show that $\psi$ increases and then
decreases. Differentiation gives
\begin{align}
\psi'(x)
=
\frac{2x\bigl[\exp(2x)(1-x)-1\bigr]}
     {(\exp(2x)-1)^2}.
\end{align}
The sign of this derivative is therefore the
sign of
\begin{align}
\exp(2x)(1-x)-1.
\end{align}
This expression is zero at $x=0$, has derivative
$\exp(2x)(1-2x)$, and equals $-1$ at $x=1$.
It increases on $(0,1/2)$ and decreases on
$(1/2,\infty)$, crossing zero exactly once
after $1/2$. Hence $\psi$ is unimodal, and its
minimum on any compact subinterval of $(0,\infty)$
occurs at an endpoint.

Since $nt>\log 2$ and $r\ge1/2$, we have
\[
nt>\frac{\log 3}{2},
\qquad
0<t\le\log 2.
\]
Put
\begin{equation}
a_0=\frac{\log 3}{2},
\qquad
k=\left\lceil\frac{a_0}{t}\right\rceil.
\end{equation}
The inequality $nt>a_0$ ensures that $1\le k\le n$.
We claim that $a_0\le kt\le2a_0$. If $t\ge a_0$,
then $k=1$ and
\[
a_0\le kt=t\le\log 2<\log 3=2a_0.
\]
If $t<a_0$, the definition of $k$ gives
\[
a_0\le kt<a_0+t<2a_0.
\]
Thus an available degree has logarithmic scale in
$[a_0,2a_0]$.

By the definition of $H$ and the unimodality of $\psi$,
\begin{align}
H\ge\psi(kt)
\ge\min\{\psi(a_0),\psi(2a_0)\}.
\end{align}
The endpoint values coincide:
\begin{align}
\psi(a_0)
=\frac{(\log 3)^2}{8}
=\psi(2a_0).
\end{align}
Consequently,
\begin{align}\label{Hbetter}
H\ge\frac{(\log 3)^2}{8}>\frac3{20}.
\end{align}
To verify the strict inequality using rational bounds,
the expansion
\begin{align}
\log\frac{1+x}{1-x}
=
2\sum_{j=0}^{\infty}\frac{x^{2j+1}}{2j+1},
\qquad |x|<1,
\end{align}
at $x=1/2$ gives
\begin{align}
\log 3>1+\frac1{12}+\frac1{80}
=\frac{263}{240}.
\end{align}
Since
\begin{align}
\left(\frac{263}{240}\right)^2-\frac65
=\frac{49}{57600}>0,
\end{align}
we obtain $(\log 3)^2/8>3/20$.

This single lower bound justifies both normalizations
used below. For $m=20/3$,
\begin{align}
mH>\frac{20}{3}\cdot\frac3{20}=1,
\end{align}
whereas for $m=55/8$,
\begin{align}
mH>\frac{55}{8}\cdot\frac3{20}
=\frac{33}{32}>1.
\end{align}
Thus the rescaled constraints in \eqref{Zconditions}
apply with either choice of $m$ throughout the
present regime.

\subsubsection{Two bounds using the coefficient masses}

Fix one of the choices of $m$ justified above.
Since $e+o=1$, the constraints in
\eqref{Zconditions} imply
\begin{align}
e\le1-d-Z,\qquad
e\le1-\sqrt{d^2+mZ^2}.
\end{align}
The first of these gives
\begin{align}
A_*
&=\frac{13e}{12}+\frac{4Z}{3}\notag\\
&\le\frac{13}{12}(1-d-Z)+\frac{4Z}{3}\notag\\
&=\frac{13}{12}(1-d)+\frac Z4.
\label{eq:outer_affine_increasing}
\end{align}

For a second bound, choose constants $a,b$ with
\begin{align}
a^2+\frac{b^2}{m}\le1.
\end{align}
Cauchy--Schwarz gives
\begin{align}
ad+bZ
&=
(a,b/\sqrt m)\cdot(d,\sqrt m Z)\notag\\
&\le
\sqrt{a^2+b^2/m}\sqrt{d^2+mZ^2}\notag\\
&\le o.
\end{align}
It follows that $e\le1-ad-bZ$, and hence
\begin{align}\label{eq:outer_affine_decreasing}
A_*
\le
\frac{13}{12}(1-ad)
+\left(\frac43-\frac{13b}{12}\right)Z.
\end{align}

We use the following two choices:
\begin{align}
\begin{array}{c|c|c|c}
m&a&b&a^2+b^2/m\\ \hline
20/3&3/5&2&24/25\\
55/8&3/4&5/3&1531/1584
\end{array}
\end{align}
Both squared norms are less than one, so both
choices are admissible. In each case,
\eqref{eq:outer_affine_increasing} increases
with $Z$, whereas
\eqref{eq:outer_affine_decreasing} decreases
with $Z$. Their minimum is therefore bounded
above by their common value at the intersection.

For $m=20/3$, the two bounds are
\begin{align}
A_*\le
\min\left\{
\frac{13}{12}(1-d)+\frac Z4,\,
\frac{13}{12}\left(1-\frac{3d}{5}\right)
-\frac{5Z}{6}
\right\}.
\end{align}
They intersect at $Z=2d/5$, giving
\begin{align}\label{eq:outer_mass_small}
A_*\le
\frac{13}{12}(1-d)+\frac d{10}
=\frac{65-59d}{60}.
\end{align}
For $m=55/8$, the bounds are
\begin{align}
A_*\le
\min\left\{
\frac{13}{12}(1-d)+\frac Z4,\,
\frac{13}{12}\left(1-\frac{3d}{4}\right)
-\frac{17Z}{36}
\right\}.
\end{align}
Their intersection is $Z=3d/8$, which gives
\begin{align}\label{eq:outer_mass_large}
A_*\le
\frac{13}{12}(1-d)+\frac{3d}{32}
=\frac{104-95d}{96}.
\end{align}
These bounds hold whether or not the intersection
point itself satisfies the remaining constraints
on $Z$.

\subsubsection{Small ceilings}



Suppose $0<\lambda\le1/5$. Choose $m=20/3$,
which is admissible by \eqref{Hbetter}.
Combining \eqref{eq:outer_coarse} with
\eqref{eq:outer_mass_small}, we obtain
\begin{align}
\frac{\mathcal I}{M}
\le
\frac{20}{3}Q(\lambda,d)\frac{65-59d}{60}.
\end{align}
The spectral cutoff \eqref{c0} already settles
$d\le c_0\lambda$. Since $c_0>2/3$, we may
therefore assume
\begin{align}
d\ge d_0:=\frac{2\lambda}{3}.
\end{align}
Both $Q(\lambda,d)$ and $65-59d$ are positive
and decreasing in $d$ on the relevant interval.
Thus
\begin{align}
\frac{\mathcal I}{M}
\le
\frac{20}{3}Q(\lambda,d_0)
\frac{65-59d_0}{60}.
\end{align}
After multiplication by the positive denominator
$60(1+d_0)(1-\lambda d_0)$, the assertion that
this last expression is at most one becomes
\begin{align}
&60(1+d_0)(1-\lambda d_0)\notag\\
&\qquad
-\frac{20}{3}\lambda(1-d_0)(65-59d_0)
\ge0.
\end{align}
Substituting $d_0=2\lambda/3$, the left side is
\begin{align}
\frac{20}{27}
\left(81-531\lambda+690\lambda^2-272\lambda^3\right).
\end{align}
The derivative of the cubic in parentheses is
\begin{align}
-531+1380\lambda-816\lambda^2
\le-531+\frac{1380}{5}
=-255
\end{align}
on $[0,1/5]$. The cubic is therefore decreasing
there, and its minimum is its value at $1/5$:
\begin{align}
81-\frac{531}{5}+\frac{690}{25}
-\frac{272}{125}
=\frac{28}{125}>0.
\end{align}
This proves $\mathcal I\le M$ throughout
$0<\lambda\le1/5$.

\subsubsection{Large ceilings}

Suppose $4/5\le\lambda<1$. Choose $m=55/8$, which is also admissible by
\eqref{Hbetter}.
Using \eqref{eq:outer_mass_large} in
\eqref{eq:outer_coarse} gives
\begin{align}
\frac{\mathcal I}{M}
\le
\frac{55}{8}Q(\lambda,d)\frac{104-95d}{96}.
\end{align}
The spectral cutoff \eqref{dstar} already
settles $d\le d_*(\lambda)$, so we may assume
$d>d_*(\lambda)$. Moreover,
\begin{align}
d_*(\lambda)-\frac{5\lambda}{6}
=
\frac{\lambda}{12}(3\lambda^2-\lambda-1)>0.
\end{align}
Indeed, $3\lambda^2-\lambda-1$ is increasing
for $\lambda\ge4/5$, and its value at $4/5$
is $3/25$. We may therefore use the simpler
lower bound
\begin{align}
d\ge d_0:=\frac{5\lambda}{6}.
\end{align}

The factors $Q(\lambda,d)$ and $104-95d$ are
positive and decreasing for $d<\lambda<1$.
It follows that
\begin{align}
\frac{\mathcal I}{M}
\le
\frac{55}{8}Q(\lambda,d_0)
\frac{104-95d_0}{96}.
\end{align}
As in the small-ceiling case, clearing the
positive denominator reduces the desired
bound to
\begin{align}
P:={}&96(1+d_0)(1-\lambda d_0)\notag\\
&-\frac{55}{8}\lambda(1-d_0)(104-95d_0)
\ge0.
\end{align}

Put $y=5\lambda-4$, so $0\le y<1$.
Substituting $\lambda=(y+4)/5$ and
$d_0=5\lambda/6$ and expanding gives
\begin{align}
1440P
=
160+17952y-10854y^2-5993y^3.
\end{align}
Since $y^2\le y$ and $y^3\le y$ for
$0\le y\le1$,
\begin{align}
1440P
&\ge
160+(17952-10854-5993)y\notag\\
&=160+1105y>0.
\end{align}
This proves $\mathcal I\le M$ for
$4/5\le\lambda<1$ and completes the treatment
of both outer ranges.

\subsection{\texorpdfstring{$u>0$, $r\ge1/2$, $nt > \log 2$:
$1/5\le\lambda\le4/5$}
{Positive slope: long logarithmic scale, central ceilings}}
\label{sec:central}

Assume $u>0$, $r\ge1/2$, $nt > \log 2$, and
$1/5\le\lambda\le4/5$. The normalization established
in Section~\ref{sec:long} gives $mH\ge1$ for
$m=55/8$. Consequently,
\begin{align}
Z=\frac{U}{\sqrt{mH}}
\end{align}
satisfies \eqref{Zconditions}, including
\begin{align}
d+Z\le o,\qquad
d^2+mZ^2\le o^2,\qquad
\frac{\mathcal I}{M}=\frac{mQZ^2}{\delta},
\end{align}
where $Q=Q(\lambda,d)$ and $\delta=\lambda-d$.

To finish this regime, it therefore suffices to prove
$mQZ^2\le\delta$. Unlike the estimates for the outer
ranges of $\lambda$, the argument below retains the
negative term in the denominator of $R$ from
Appendix~\ref{sec:forward}. This term provides the
additional control needed in the central range.

\begin{proposition}\label{central}
Retain the polynomial assumptions and coefficient
masses of Section~\ref{sec:coverage}, together with
the differential estimate of
Appendix~\ref{sec:forward}. Let $m=55/8$, and suppose
\begin{align}
\frac15\le\lambda\le\frac45,\qquad
\frac{37\lambda}{50}\le d<\lambda.
\end{align}
If $0<Z\le U$ satisfies
\begin{align}
d+Z\le o,\qquad d^2+mZ^2\le o^2,
\end{align}
then
\begin{align}
mQ(\lambda,d)Z^2\le\lambda-d.
\end{align}
\end{proposition}

\begin{proof}
The proof has three parts. We first assume that the
claimed inequality fails and derive a necessary
inequality for the denominator of $R$. We then
bound that denominator by two affine expressions.
Finally, we show that at least one of these
expressions contradicts the necessary inequality.

\paragraph{Step 1: reduce a putative violation to a denominator bound}

Suppose, for a contradiction, that
$mQZ^2>\delta$. Define
\begin{align}
z=\sqrt{\frac{\delta}{mQ}},
\end{align}
so that $0<z<Z\le U$. Here $z$ is a local
scaled-slope variable, not the joint output
probability used in the Boolean reduction.

By Proposition~\ref{riccati},
$U\le o/\mathrm e$, the denominator of $R$
is positive on $[0,o/\mathrm e]$, and $R$
is strictly increasing there. Hence
\begin{align}
\delta\ge R(U)>R(z).
\end{align}
Writing
\begin{align}
D=\frac{13e}{12}+\frac{4z}{3}
-\frac{z^2}{o-4z/5},
\end{align}
we therefore have $D>0$ and
\begin{align}\label{contradiction}
\delta>\frac{z^2}{D},
\qquad\text{so}\qquad
mQD>1.
\end{align}
The last implication uses $z^2=\delta/(mQ)$
and $\delta>0$.

Since $z<Z$, the constraints on $Z$ also imply
\begin{align}
o\ge d+z,\qquad
o\ge\sqrt{d^2+mz^2}.
\label{eq:central_z_constraints}
\end{align}
We will use these inequalities, together with
$e=1-o$, to bound $D$ from above.

\paragraph{Step 2: two affine upper bounds on the denominator}

We claim that
\begin{align}\label{planes}
D\le
\min\left\{
\frac{13}{12}-d-\frac{3z}{10},\,
\frac{13}{12}-\frac{2d}{3}-\frac{5z}{4}
\right\}.
\end{align}
For this calculation, put
\begin{align}
C=o-\frac{4z}{5}>0.
\end{align}
Its positivity follows, for example, from
$z\le U\le o/\mathrm e<o$.

For the first bound, write
\begin{align}
D
&=\frac{13}{12}-\frac{13o}{12}
+\frac{4z}{3}-\frac{z^2}{C}\notag\\
&=\frac{13}{12}-o+\frac{4z}{3}
-\left(\frac{o}{12}+\frac{z^2}{C}\right)\notag\\
&\le
\frac{13}{12}-d+\frac z3
-\left(\frac{o}{12}+\frac{z^2}{C}\right),
\end{align}
where the last line uses $o\ge d+z$.
Since $o=C+4z/5$,
\begin{align}
\frac{o}{12}+\frac{z^2}{C}
=
\frac{C}{12}+\frac z{15}+\frac{z^2}{C}.
\end{align}
The arithmetic--geometric mean inequality gives
\begin{align}
\frac{C}{12}+\frac{z^2}{C}
\ge2\sqrt{\frac{z^2}{12}}
=\frac z{\sqrt3}.
\end{align}
Also $1/\sqrt3\ge17/30$, as follows by squaring.
Therefore
\begin{align}
\frac{o}{12}+\frac{z^2}{C}
\ge
\left(\frac{17}{30}+\frac1{15}\right)z
=\frac{19z}{30}.
\end{align}
Substitution yields
\begin{align}
D\le\frac{13}{12}-d-\frac{3z}{10}.
\end{align}

For the second bound, use
$(z-C/3)^2\ge0$, which implies
\begin{align}
\frac{z^2}{C}\ge\frac{2z}{3}-\frac C9.
\end{align}
Consequently,
\begin{align}
D
&\le
\frac{13}{12}-\frac{13o}{12}
+\frac{4z}{3}-\frac{2z}{3}+\frac C9\notag\\
&=
\frac{13}{12}-\frac{35o}{36}+\frac{26z}{45}.
\label{eq:central_second_plane_preliminary}
\end{align}
We next obtain a linear lower bound on $o$
from the second constraint in
\eqref{eq:central_z_constraints}. Cauchy--Schwarz gives
\begin{align}
\frac{24d}{35}+\frac{329z}{175}
&=
\left(\frac{24}{35},\frac{329}{175\sqrt m}\right)
\cdot(d,\sqrt m z)\notag\\
&\le
\left[
\left(\frac{24}{35}\right)^2
+\frac1m\left(\frac{329}{175}\right)^2
\right]^{1/2}
\sqrt{d^2+mz^2}\notag\\
&\le o,
\end{align}
because
\begin{align}
\left(\frac{24}{35}\right)^2
+\frac1m\left(\frac{329}{175}\right)^2
=
\frac{1657928}{1684375}<1.
\end{align}
The coefficient of $o$ in
\eqref{eq:central_second_plane_preliminary}
is negative, so this lower bound on $o$ gives
an upper bound on $D$. Using
\begin{align}
\frac{35}{36}\frac{24}{35}=\frac23,\qquad
-\frac{35}{36}\frac{329}{175}
+\frac{26}{45}=-\frac54,
\end{align}
we obtain
\begin{align}
D\le\frac{13}{12}-\frac{2d}{3}-\frac{5z}{4}.
\end{align}
This proves \eqref{planes}.

\paragraph{Step 3: formulate the remaining scalar inequalities}

Write
\begin{align}
y=1-\frac d\lambda,\qquad
d=\lambda(1-y),\qquad
\delta=\lambda y.
\end{align}
The assumed lower bound on $d$ gives
$0<y\le13/50$. Put $\kappa=13/12$.
For either of the coefficient pairs
\begin{align}
(a,b)=(1,3/10)
\qquad\text{or}\qquad
(a,b)=(2/3,5/4),
\end{align}
define
\begin{align}\label{eq:central_G}
G_{a,b}(\lambda,y)
=
1-mQ(\lambda,d)(\kappa-ad)
+b\sqrt{mQ(\lambda,d)(\lambda-d)},
\qquad d=\lambda(1-y).
\end{align}
The letters $a,b$ in this calculation are
coefficients of the affine bounds in
\eqref{planes}, not the marginal endpoint
parameters used earlier.

To see why these functions are useful, recall that
$z=\sqrt{\delta/(mQ)}$. Each corresponding bound
in \eqref{planes} gives
\begin{align}
mQD
&\le mQ(\kappa-ad-bz)\notag\\
&=mQ(\kappa-ad)-b\sqrt{mQ\delta}\notag\\
&=1-G_{a,b}(\lambda,y).
\end{align}
Thus $G_{a,b}(\lambda,y)\ge0$ would imply
$mQD\le1$, contradicting \eqref{contradiction}.

We will establish
\begin{align}
&G_{1,3/10}(\lambda,y)\ge0,
&&0\le y\le1/16,
\label{eq:central_G_first}\\
&G_{2/3,5/4}(\lambda,y)\ge0,
&&1/16\le y\le13/50,
\label{eq:central_G_second}
\end{align}
uniformly for $1/5\le\lambda\le4/5$.
The value at $y=0$ is included by continuity.
These two intervals cover every possible
relative gap in the proposition.

\subsubsection*{Concavity in the relative gap}

Fix $\lambda\in[1/5,4/5]$. We show that each
function in \eqref{eq:central_G_first}--%
\eqref{eq:central_G_second} is concave in $y$
on its stated interval. Since
$d=\lambda(1-y)$ is affine in $y$, it suffices
to establish the corresponding concavity in $d$.

There are two nonconstant terms in
\eqref{eq:central_G}. We prove that
$\sqrt{(\lambda-d)Q(\lambda,d)}$ is concave,
and that $Q(\lambda,d)(\kappa-ad)$ is convex.
The signs in \eqref{eq:central_G} then give
the required concavity of $G$.

\paragraph{The square-root term}

For this calculation, put
\begin{align}
v(d)=(\lambda-d)Q(\lambda,d)
=
\frac{\lambda(\lambda-d)(1-d)}
{(1+d)(1-\lambda d)}.
\end{align}
The symbol $v(d)$ here denotes an auxiliary
function, rather than a correlation parameter.
For $d<\lambda$, this function is positive and
\begin{align}
(\sqrt v)''
=
\frac{2vv''-(v')^2}{4v^{3/2}}.
\end{align}
We therefore need to show
$2vv''-(v')^2\le0$.

Introduce the algebraic variable
\begin{align}
\eta=\frac{\lambda-d}{1-\lambda d},
\qquad
\lambda=\frac{d+\eta}{1+d\eta}.
\end{align}
This $\eta$ is local to the present calculation
and is distinct from the second endpoint in
the Boolean normalization. Differentiating
$v(d)$ with $\lambda$ fixed and then making
the displayed substitution gives
\begin{align}
2vv''-(v')^2
=
\frac{(d+\eta)^2}
{(1+d)^4(1+d\eta)^2}\,J(d,\eta),
\end{align}
where
\begin{align}
J(d,\eta)
=
-1+4\eta(1-d)+(2-8d)\eta^2
-4(1-d)\eta^3+3\eta^4.
\end{align}
In this identity, the substitution is made
after differentiation; $\eta$ is not held
fixed when differentiating $v$.

We first bound the last two terms of $J$.
The inequality $\lambda\le4/5$ implies
$\eta\le1-d$. Indeed,
\begin{align}
\eta\le1-d
&\quad\Longleftrightarrow\quad
\lambda-d\le(1-d)(1-\lambda d)\notag\\
&\quad\Longleftrightarrow\quad
\lambda(1+d-d^2)\le1,
\end{align}
and the last inequality follows from
$1+d-d^2\le5/4$. Consequently,
\begin{align}
-4(1-d)\eta^3+3\eta^4
=
\eta^3[3\eta-4(1-d)]
\le-\eta^3(1-d)\le0.
\end{align}

If $d\ge1/4$, the coefficient $2-8d$ is
nonpositive. Also,
\begin{align}
\eta(1-d)
=
(\lambda-d)\frac{1-d}{1-\lambda d}
\le\lambda-d,
\end{align}
because $\lambda\le1$. It follows that
\begin{align}
J(d,\eta)
&\le-1+4\eta(1-d)\notag\\
&\le-1+4(\lambda-d)\notag\\
&\le-1+4\cdot\frac{13}{50}\cdot\frac45
=-\frac{21}{125}<0.
\end{align}

If $d\le1/4$, the assumption
$d\ge37\lambda/50$ implies
\begin{align}
\lambda\le\frac{50d}{37}\le\frac{25}{74}.
\end{align}
Hence
\begin{align}
\lambda-d\le\frac{13\lambda}{50}
\le\frac{13}{148},
\qquad
\lambda d\le\frac{25}{296}.
\end{align}
Therefore
\begin{align}
\eta
\le
\frac{13/148}{1-25/296}
=\frac{26}{271}<\frac18.
\end{align}
Dropping the nonpositive cubic and quartic
combination and using $1-d\le1$ and
$2-8d\le2$, we obtain
\begin{align}
J(d,\eta)
\le-1+4\eta+2\eta^2
<
-1+\frac12+\frac1{32}
=-\frac{15}{32}<0.
\end{align}
Thus $\sqrt{v(d)}$ is concave throughout
$37\lambda/50\le d<\lambda$. Since it extends
continuously to zero at $d=\lambda$, concavity
also holds on the corresponding closed interval.

\paragraph{The rational term}

The partial-fraction identity
\begin{align}
Q(\lambda,d)(\kappa-ad)
={}&-a
+\frac{2\lambda(a+\kappa)}
{(1+\lambda)(1+d)}\notag\\
&+\frac{(a-\kappa\lambda)(1-\lambda)}
{(1+\lambda)(1-\lambda d)}
\label{eq:central_partial_fractions}
\end{align}
makes the second derivative easy to analyze.
One way to verify this identity is to multiply
both sides by $(1+d)(1-\lambda d)$.
Comparison of the quadratic coefficients gives
the constant term $-a$; evaluation of the resulting
polynomial identity at $d=-1$ and $d=1/\lambda$
determines the other two coefficients.

Both $1/(1+d)$ and $1/(1-\lambda d)$ have positive
second derivatives on the present domain.
For $a=1$, both coefficients in
\eqref{eq:central_partial_fractions} are
nonnegative, because
\begin{align}
1-\kappa\lambda
\ge1-\frac{13}{12}\frac45
=\frac2{15}>0.
\end{align}
Thus the rational term is convex for $a=1$.

For $a=2/3$, the same argument applies when
$\lambda\le a/\kappa=8/13$.
Suppose now that $8/13\le\lambda\le4/5$.
The second coefficient can be negative, so
we compare the two contributions to the
second derivative explicitly. From
\eqref{eq:central_partial_fractions}, convexity
is equivalent to
\begin{align}
(\kappa\lambda-a)(1-\lambda)\lambda
\left(\frac{1+d}{1-\lambda d}\right)^3
\le2(a+\kappa).
\label{eq:central_rational_convexity}
\end{align}
For this choice of $a$, we only need the range
$y\ge1/16$. Therefore
\begin{align}
d=\lambda(1-y)
\le\frac45\frac{15}{16}=\frac34,
\end{align}
and
\begin{align}
\frac{1+d}{1-\lambda d}
\le
\frac{7/4}{1-(4/5)(3/4)}
=\frac{35}{8}.
\end{align}

The remaining factor satisfies
\begin{align}
(\kappa\lambda-a)(1-\lambda)\lambda
\le\frac4{125}
\qquad(8/13\le\lambda\le4/5).
\end{align}
To verify this, its derivative is
\begin{align}
\frac{-39\lambda^2+42\lambda-8}{12}.
\end{align}
The numerator is a concave quadratic, positive
at both endpoints of the interval, and hence
positive throughout. The factor is therefore
increasing, and its maximum is its value
$4/125$ at $\lambda=4/5$.

Combining the two estimates gives
\begin{align}
(\kappa\lambda-a)(1-\lambda)\lambda
\left(\frac{1+d}{1-\lambda d}\right)^3
\le
\frac4{125}\left(\frac{35}{8}\right)^3
=\frac{343}{128}
<\frac72=2(a+\kappa).
\end{align}
This proves \eqref{eq:central_rational_convexity}
and hence the required convexity for $a=2/3$.

The square-root term in $G$ is therefore concave,
and the rational term multiplied by $-m$ is
also concave. Since $d$ is affine in $y$, both
functions $G$ are concave in $y$ on the intervals
specified in \eqref{eq:central_G_first} and
\eqref{eq:central_G_second}.

\subsubsection*{Reduction to four endpoint values}

A concave function on a closed interval lies
above the line segment joining its endpoint
values. It is therefore enough to prove
\begin{align}
G_{1,3/10}(\lambda,0)\ge0,\qquad
G_{1,3/10}(\lambda,1/16)\ge0,
\end{align}
and
\begin{align}
G_{2/3,5/4}(\lambda,1/16)\ge0,\qquad
G_{2/3,5/4}(\lambda,13/50)\ge0.
\end{align}

At $y=0$, we have $d=\lambda$ and
\begin{align}
Q(\lambda,\lambda)=\frac{\lambda}{(1+\lambda)^2}.
\end{align}
The square-root term vanishes, giving
\begin{align}
G_{1,3/10}(\lambda,0)
=
\frac{756\lambda^2-523\lambda+96}
{96(1+\lambda)^2}>0.
\end{align}
The numerator has positive leading coefficient
and discriminant
\begin{align}
523^2-4\cdot756\cdot96=-16775<0.
\end{align}

For the other three endpoint values, put
\begin{align}
s=1-y,\qquad
K=(1+s\lambda)(1-s\lambda^2)>0.
\end{align}
These are local algebraic quantities.
Since $d=s\lambda$,
\begin{align}
Q(\lambda,d)=\frac{\lambda(1-s\lambda)}{K},
\qquad
\lambda-d=\lambda y.
\end{align}
Multiplying \eqref{eq:central_G} by $K$ gives
\begin{align}
KG_{a,b}
={}&K-m\lambda(1-s\lambda)(\kappa-as\lambda)
\notag\\
&+b\lambda\sqrt{my}\,
\sqrt{(1-s^2\lambda^2)(1-s\lambda^2)}.
\label{eq:central_KG}
\end{align}
For the square-root term, this follows from
\begin{align}
K\sqrt{\frac{1-s\lambda}{K}}
=
\sqrt{(1-s\lambda)K}
=
\sqrt{(1-s^2\lambda^2)(1-s\lambda^2)}.
\end{align}

We use the following rational lower bounds:
\begin{align}
\begin{array}{c|c|c|c}
y&s&\sqrt{my}\ge v_0&
\sqrt{(1-s^2\lambda^2)(1-s\lambda^2)}
 \ge1-k\lambda^2\\ \hline
1/16&15/16&13/20&k=11/12\\
13/50&37/50&4/3&k=13/20
\end{array}
\end{align}
Here $v_0$ and $k$ are auxiliary constants for
these two endpoint calculations.
The first square-root bounds follow from
\begin{align}
\frac{55}{128}>\frac{169}{400},
\qquad
\frac{143}{80}>\frac{16}{9}.
\end{align}
To verify the other bounds, first note that
$1-k\lambda^2>0$ in both rows. Squaring is
therefore legitimate. The difference of the
squared expressions is
\begin{align}
&(1-s^2\lambda^2)(1-s\lambda^2)
-(1-k\lambda^2)^2\notag\\
&=
\lambda^2\bigl[
(2k-s-s^2)+(s^3-k^2)\lambda^2
\bigr].
\end{align}
In each row, the constant in brackets is
positive and the coefficient of $\lambda^2$
is negative. Its minimum on the present range
therefore occurs at $\lambda^2=16/25$.
The resulting values are
\begin{align}
\frac{187}{28800}>0
\quad\text{and}\quad
\frac{2099}{1562500}>0,
\end{align}
respectively.

It follows from \eqref{eq:central_KG} that
$KG_{a,b}$ is bounded below by the cubic
\begin{align}
P(\lambda)
=
K-m\lambda(1-s\lambda)(\kappa-as\lambda)
+bv_0\lambda(1-k\lambda^2).
\end{align}
For the three endpoint choices, these cubics are
\begin{align}
P_1(\lambda)
&=
1-\frac{970048}{153600}\lambda
+\frac{1918500}{153600}\lambda^2
-\frac{1090581}{153600}\lambda^3,\\
P_2(\lambda)
&=
1-\frac{17504}{3072}\lambda
+\frac{31770}{3072}\lambda^2
-\frac{17363}{3072}\lambda^3,\\
P_3(\lambda)
&=
1-\frac{604950}{120000}\lambda
+\frac{979575}{120000}\lambda^2
-\frac{496892}{120000}\lambda^3.
\end{align}
Here the corresponding triples $(a,b,y)$ are
\begin{align}
(1,3/10,1/16),\qquad
(2/3,5/4,1/16),\qquad
(2/3,5/4,13/50),
\end{align}
respectively. Thus it remains to verify
positivity of these three cubics on
$[1/5,4/5]$.

\subsubsection*{Elementary verification of cubic positivity}

For each cubic, we expand about a point near
its minimum and absorb the cubic term into
the quadratic term. The resulting quadratic
lower bounds are positive by completion of
the square.

For $P_1$, put $w=\lambda-3/8$. Exact expansion
gives
\begin{align}
P_1(\lambda)
={}&
\frac{360099}{26214400}
+\frac{559241}{9830400}w
+\frac{1844257}{409600}w^2
-\frac{363527}{51200}w^3.
\end{align}
Since $w\le4/5-3/8=17/40$,
\begin{align}
\frac{1844257}{409600}w^2
-\frac{363527}{51200}w^3
&=
w^2\left(
\frac{1844257}{409600}
-\frac{363527}{51200}w
\right)\notag\\
&\ge\frac{1520663}{1024000}w^2
\ge w^2.
\end{align}
The constant coefficient is greater than
$1/75$, and the absolute linear coefficient
is less than $3/50$. Hence
\begin{align}
P_1(\lambda)
\ge w^2-\frac3{50}|w|+\frac1{75}>0.
\end{align}

For $P_2$, put $w=\lambda-2/5$. Then
\begin{align}
P_2(\lambda)
=
\frac{331}{24000}
-\frac{2639}{19200}w
+\frac{1139}{320}w^2
-\frac{17363}{3072}w^3.
\end{align}
Using $w\le2/5$, we obtain
\begin{align}
\frac{1139}{320}
-\frac{17363}{3072}w
\ge\frac{9973}{7680}>1.
\end{align}
Also $331/24000>1/75$ and
$2639/19200<1/7$. Therefore
\begin{align}
P_2(\lambda)
\ge w^2-\frac17|w|+\frac1{75}>0.
\end{align}

For $P_3$, put $w=\lambda-1/2$. The expansion is
\begin{align}
P_3(\lambda)
=
\frac{1229}{480000}
+\frac{163}{10000}w
+\frac{78079}{40000}w^2
-\frac{124223}{30000}w^3.
\end{align}
Since $w\le3/10$,
\begin{align}
\frac{78079}{40000}
-\frac{124223}{30000}w
\ge\frac{141949}{200000}>\frac7{10}.
\end{align}
The constant coefficient exceeds $1/400$,
and $163/10000<1/60$, giving
\begin{align}
P_3(\lambda)
\ge\frac7{10}w^2-\frac1{60}|w|+\frac1{400}>0.
\end{align}

For completeness, all three final inequalities
follow from
\begin{align}
Aw^2-B|w|+C
=
A\left(|w|-\frac{B}{2A}\right)^2
+C-\frac{B^2}{4A},
\end{align}
with $A>0$ and $C-B^2/(4A)>0$.
The corresponding positive remainders are
\begin{align}
\frac1{75}-\frac9{10000},\qquad
\frac1{75}-\frac1{196},\qquad
\frac1{400}-\frac1{10080}.
\end{align}

Thus all four required endpoint values of $G$
are positive. Concavity then proves
\eqref{eq:central_G_first} and
\eqref{eq:central_G_second} throughout their
respective intervals. For every admissible
$y$, at least one bound in \eqref{planes}
therefore gives $mQD\le1$, contradicting
\eqref{contradiction}. The proposition follows.
\end{proof}

We now apply the proposition to the normalized
slope $Z=U/\sqrt{mH}$. The spectral cutoff
\eqref{c0} already settles $d\le741\lambda/1000$.
At every remaining point,
\begin{align}
d>\frac{741\lambda}{1000}
>\frac{37\lambda}{50},
\end{align}
and \eqref{Zconditions} supplies all the other
hypotheses of Proposition~\ref{central}.
Consequently,
\begin{align}
\frac{\mathcal I}{M}
\le\frac{mQZ^2}{\delta}\le1.
\end{align}
This proves the scalar inequality in the central
range.  

Sections~\ref{sec:small}, \ref{sec:long},
\ref{sec:central}, and \ref{sec:short} cover
all positive slopes, as explained in
Section~\ref{sec:coverage}.
Section~\ref{sec:decreasing} covers negative
slopes, and the zero-slope case is immediate.
This proves Lemma~\ref{lem:envelope}. The reduction
in Section~\ref{sec:main} then establishes
Theorem~\ref{main}.

\section{Conclusion}\label{sec:conclusion}


We have determined the maximum Fisher information
retained by two arbitrary Boolean encoders of correlated
uniform binary vectors, thereby proving the
Amari--Kobayashi conjecture. Matching parity functions
attain the optimum, with their degree chosen to maximize
the parity envelope. The dependence of the optimal
degree on the correlation reflects a balance between
the magnitude of the output correlation and its
sensitivity to changes in the input correlation.

The proof reduces the Fisher-information comparison to a
scalar inequality for the normalized covariance polynomial.
A factorization of the four-cell curvature, together with
the Fourier coefficient bound and the global range
constraint, yields the sharp parity benchmark. Although
the proof is entirely analytical, the scalar argument
still requires a substantial division into parameter
regimes. Finding a shorter proof of this inequality, or
a more conceptual proof of the main theorem, remains an
important direction. Such a proof could clarify how the
coefficient and range constraints enforce parity optimality
and which aspects of the argument extend to other
information measures.

The perspective developed in \cite{chen} suggests
that the Amari--Kobayashi and Courtade--Kumar conjectures
may be extremal instances of a broader family of
relative-entropy inequalities for Boolean functions.
In both problems, a Fourier character attains the
optimum, suggesting that admissible combinations of
noise eigenmodes cannot outperform an appropriate
single mode. Convexity provides a direct explanation
when the Fourier representation yields a mixture of
character-induced output laws with weights independent
of the correlation parameter; more generally, the
challenge is to understand how the constraints imposed
by Boolean functions enforce this extremal behavior.
Fisher information describes the second-order behavior
of relative entropy between nearby parameter values,
so the present result provides a local constraint
that any such family must respect. A local inequality
does not, by itself, establish the corresponding
comparison at finite parameter separations.
Identifying the appropriate family, and determining
how convexity and spectral constraints connect its
local and finite-separation forms, are therefore
natural next questions.

More broadly, Fourier methods for information inequalities
on the discrete cube can be understood as an optimization
relaxation: the constraints imposed by Boolean functions
are replaced by tractable conditions on their Fourier
coefficients or associated polynomials. The effectiveness
of this approach depends on retaining the constraints
that matter for the information functional under
consideration. Here, after bounding the four-cell curvature,
the absolute coefficient budget and the global range
constraint suffice to recover the sharp optimum.
Thus the relaxation preserves the optimum even though
it admits polynomials that need not arise from Boolean
functions. Understanding which spectral and range
constraints must be retained, and which can be discarded
without changing the optimum, may provide a systematic
route to further sharp information inequalities for
Boolean functions.

\appendices
\section{A Differential Estimate Along the Logarithmic Path}
\label{sec:forward}

This appendix establishes an auxiliary estimate used in
Sections~\ref{sec:long} and~\ref{sec:central}. It applies
under the polynomial assumptions of Lemma~\ref{lem:envelope}
at every evaluation point $r\in(0,1)$ with
$u=h'(r)>0$. No additional restriction on $r$ or
$n(-\log r)$ is imposed. The argument uses only the
polynomial assumptions and does not invoke the
conclusion of Lemma~\ref{lem:envelope}.

Retain the decomposition and notation from
Section~\ref{sec:scalar}:
\begin{align}
h=B-A,\qquad
e=A(1),\qquad o=B(1),\qquad e+o=1,
\end{align}
where $A$ and $B$ have nonnegative coefficients.
Also retain
\begin{align}
d=h(r),\qquad
\delta=\lambda-d,\qquad
F(v)=\lambda-h(v)\ge0\quad(0\le v\le1).
\end{align}
The mass bounds \eqref{massbounds} imply $e>0$
and $o>0$. Recall the logarithmic normalization
\begin{align}
t=-\log r,\qquad
U=rtu,\qquad
H=Mr^2t^2
=\max_{1\le k\le n}
\frac{(kt)^2}{\exp(2kt)-1},
\end{align}
where $M=M_n(r)$.

The path $v=r^{1-\theta}$ moves from $r$ toward
$1$ as $\theta$ increases from zero to one.
Along this path, each monomial becomes
$\exp(-kt(1-\theta))$. This representation permits
bounds that depend on the coefficient masses
rather than on individual degrees. The argument
uses nonnegativity of $F$ along the entire path
to relate the initial slope $U$ to the initial
gap $\delta$.


\begin{proposition}\label{riccati}
Let $c=13/24$. With the notation above,
\begin{align}\label{Rlemma}
\delta\ge R(U):=
\frac{U^2}
{13e/12+4U/3-U^2/(o-4U/5)}.
\end{align}
Here $0<U\le o/\mathrm e$. The denominator is
positive on $[0,o/\mathrm e]$, and $R$, with
$R(0)=0$, is strictly increasing on this interval.
\end{proposition}

\begin{proof}
We first derive bounds along the path, then
compare the normalized slope with the solution of a scalar differential equation. Integrating
this comparison gives an integral lower bound
on $\delta$, from which \eqref{Rlemma} follows.

\paragraph{Step 1: estimates along the logarithmic path}

For $0\le\theta<1$, define
\begin{align}
a(\theta)=A(r^{1-\theta}),\qquad
b(\theta)=B(r^{1-\theta}),
\end{align}
and
\begin{align}
\mathcal F(\theta)
=
\lambda+a(\theta)-b(\theta)
=
F(r^{1-\theta})\ge0.
\end{align}
The functions $a$ and $b$ here denote the
compositions along the path; they are distinct
from the marginal endpoint constants used in
the Boolean normalization.

At $\theta=0$, we have $\mathcal F(0)=\delta$.
Moreover,
\begin{align}
\left.\frac{d}{d\theta}r^{1-\theta}
\right|_{\theta=0}
=rt,
\end{align}
so the chain rule gives
\begin{align}
-\mathcal F'(0)=rt\,h'(r)=U.
\end{align}

For each degree $k$, put $x_k=kt(1-\theta)$.
Since $r^{k(1-\theta)}=\exp(-x_k)$,
\begin{align}
&\frac{d}{d\theta}r^{k(1-\theta)}
=kt\exp(-x_k),\\
&\frac{d^2}{d\theta^2}r^{k(1-\theta)}
=(kt)^2\exp(-x_k).
\end{align}
Consequently,
\begin{align}
&(1-\theta)^2a''(\theta)
=\sum_k\alpha_kx_k^2\exp(-x_k),\\
&b(\theta)+(1-\theta)b'(\theta)
=\sum_k\beta_k(1+x_k)\exp(-x_k),\\
&(1-\theta)b'(\theta)
=\sum_k\beta_kx_k\exp(-x_k).
\end{align}
For $x\ge0$, elementary differentiation gives
\begin{align}
x^2\exp(-x)\le\frac4{\mathrm e^2}<c,\qquad
(1+x)\exp(-x)\le1,\qquad
x\exp(-x)\le\frac1{\mathrm e}.
\end{align}
Using $\sum_k\alpha_k=e$ and $\sum_k\beta_k=o$
therefore yields
\begin{align}
&(1-\theta)^2a''(\theta)\le ce,
\label{eq:path_a_bound}\\
&b(\theta)+(1-\theta)b'(\theta)\le o,
\label{eq:path_b_budget}\\
&(1-\theta)b'(\theta)\le o/\mathrm e.
\label{eq:path_b_slope}
\end{align}

We also need a lower bound on $b''$. By
Cauchy--Schwarz,
\begin{align}
b'(\theta)^2
&=
\left(\sum_k\beta_kkt\exp(-x_k)\right)^2\notag\\
&\le
\left(\sum_k\beta_k\exp(-x_k)\right)
\left(\sum_k\beta_k(kt)^2\exp(-x_k)\right)\notag\\
&=b(\theta)b''(\theta).
\end{align}
Since $o>0$ and $r^{1-\theta}>0$, we have
$b(\theta)>0$, and hence
\begin{align}\label{eq:path_b_curvature}
b''(\theta)\ge\frac{b'(\theta)^2}{b(\theta)}.
\end{align}

In particular, $a'(0)\ge0$ and
$U=b'(0)-a'(0)>0$ imply
\begin{align}
0<U\le b'(0)\le o/\mathrm e.
\end{align}

\paragraph{Step 2: a differential inequality for the slope}

Introduce the new path parameter
\begin{align}
\xi=-\log(1-\theta),\qquad
\theta=1-\exp(-\xi),\qquad \xi\ge0,
\end{align}
and define
\begin{align}
W(\xi)
=
-\frac{(1-\theta)\mathcal F'(\theta)}{o}
=
\frac{(1-\theta)(b'(\theta)-a'(\theta))}{o},
\qquad
q=\frac{ce}{o}>0.
\end{align}
The symbol $q$ in this proof is a local scalar
constant, not the output marginal probability
used in the Boolean reduction.
Since $d\theta/d\xi=1-\theta$,
\begin{align}
W(0)=\frac Uo,\qquad
\frac{d}{d\xi}\mathcal F(\theta(\xi))=-oW(\xi).
\label{eq:path_gap_derivative}
\end{align}

We derive a differential inequality on any
interval where $W>0$. Because $a'\ge0$,
\begin{align}
oW\le(1-\theta)b'\le o/\mathrm e.
\end{align}
Thus $0<W\le1/\mathrm e$. Also,
\eqref{eq:path_b_budget} gives
\begin{align}
b\le o-(1-\theta)b'\le o(1-W).
\end{align}
Using \eqref{eq:path_b_curvature}, we obtain
\begin{align}
\frac{(1-\theta)^2}{o}b''
&\ge
\frac{((1-\theta)b')^2}{ob}\notag\\
&\ge
\frac{o^2W^2}{o^2(1-W)}
=\frac{W^2}{1-W}.
\end{align}
All denominators are positive because
$W\le1/\mathrm e<1$.

Differentiating the definition of $W$ gives
\begin{align}
\frac{dW}{d\xi}
=
-W+\frac{(1-\theta)^2}{o}(b''-a'').
\end{align}
Combining the preceding estimates with
\eqref{eq:path_a_bound} yields
\begin{align}\label{eq:W_differential}
\frac{dW}{d\xi}
\ge
-W+\frac{W^2}{1-W}-q
=-D_0(W),
\end{align}
where
\begin{align}
D_0(w)=q+w-\frac{w^2}{1-w}.
\end{align}
For $0\le w\le1/\mathrm e<1/2$,
\begin{align}
D_0(w)
=
q+\frac{w(1-2w)}{1-w}
\ge q>0.
\label{eq:D0_positive}
\end{align}

\paragraph{Step 3: comparison and integration}

Put $s=U/o$, and consider the solution of
\begin{align}
\frac{dw}{d\xi}=-D_0(w),\qquad w(0)=s.
\end{align}
Here $s$ is a local scalar denoting the initial
normalized slope, not the covariance polynomial.
Because $D_0$ is continuously differentiable
on $[0,1/\mathrm e]$ and strictly positive
there, this solution decreases until it reaches
zero. Its hitting time is
\begin{align}
T=\int_0^s\frac{dw}{D_0(w)}
\le\frac{s}{q}<\infty.
\end{align}

The scalar comparison theorem applied to
\eqref{eq:W_differential} gives $W(\xi)\ge w(\xi)$
as long as both functions remain positive.
This suffices to show that $W$ stays positive
throughout $[0,T)$. Indeed, if $W$ first reached
zero at a time $\tau<T$, comparison up to $\tau$
and continuity would imply
\begin{align}
0=W(\tau)\ge w(\tau)>0,
\end{align}
a contradiction. Hence
\begin{align}
W(\xi)\ge w(\xi)>0\qquad(0\le\xi<T).
\end{align}
No assertion about the sign of $W$ after $T$
is needed.

Integrating \eqref{eq:path_gap_derivative} up
to $T$ gives
\begin{align}
\mathcal F(\theta(T))
=
\delta-o\int_0^T W(\xi)\,d\xi.
\end{align}
The left side is nonnegative by the global
constraint. Therefore
\begin{align}
\delta
&\ge o\int_0^T W(\xi)\,d\xi\notag\\
&\ge o\int_0^T w(\xi)\,d\xi\notag\\
&=o\int_0^s\frac{w}{D_0(w)}\,dw.
\end{align}
In the last step, we used
$d\xi=-dw/D_0(w)$ along the comparison solution.
Thus
\begin{align}\label{integral}
\delta\ge
o\int_0^{U/o}
\frac{w}{q+w-w^2/(1-w)}\,dw.
\end{align}

\paragraph{Step 4: an explicit lower bound on the integral}

Apply Cauchy--Schwarz to
$\sqrt{w/D_0(w)}$ and $\sqrt{wD_0(w)}$ on $[0,s]$.
Since $D_0(w)>0$,
\begin{align}
\left(\int_0^s w\,dw\right)^2
\le
\left(\int_0^s\frac{w}{D_0(w)}\,dw\right)
\left(\int_0^s wD_0(w)\,dw\right).
\end{align}
Consequently,
\begin{align}\label{eq:path_integral_cs}
\int_0^s\frac{w}{D_0(w)}\,dw
\ge
\frac{s^4/4}
{qs^2/2+s^3/3-\displaystyle\int_0^s
 \frac{w^3}{1-w}\,dw}.
\end{align}
The denominator equals
$\int_0^s wD_0(w)\,dw$ and is positive.

To bound the remaining integral, use the probability
density $4w^3/s^4$ on $[0,s]$. Its mean is
\begin{align}
\int_0^s w\,\frac{4w^3}{s^4}\,dw=\frac{4s}{5}.
\end{align}
The function $w\mapsto1/(1-w)$ is convex on
$[0,s]$, since $s\le1/\mathrm e<1$.
Jensen's inequality therefore gives
\begin{align}
\int_0^s\frac{w^3}{1-w}\,dw
&=
\frac{s^4}{4}
\int_0^s\frac{1}{1-w}\frac{4w^3}{s^4}\,dw\notag\\
&\ge
\frac{s^4}{4(1-4s/5)}.
\label{eq:path_integral_jensen}
\end{align}

The integral in \eqref{eq:path_integral_cs}
is subtracted from the denominator. Replacing
it by the lower bound
\eqref{eq:path_integral_jensen} therefore enlarges
that denominator and gives a weaker, explicit
lower bound. In particular, the enlarged
denominator remains positive. We obtain
\begin{align}
\int_0^s\frac{w}{D_0(w)}\,dw
&\ge
\frac{s^4/4}
{qs^2/2+s^3/3-s^4/[4(1-4s/5)]}\notag\\
&=
\frac{s^2}{2q+4s/3-s^2/(1-4s/5)}.
\end{align}
Substituting $s=U/o$ and $q=ce/o$ into
\eqref{integral}, and multiplying numerator
and denominator by $o$, gives
\begin{align}
\delta
\ge
\frac{U^2}
{2ce+4U/3-U^2/(o-4U/5)}.
\end{align}
Since $2c=13/12$, this is \eqref{Rlemma}.

\paragraph{Step 5: positivity and monotonicity}

Write the denominator of $R(U)$ as
\begin{align}
D(U)=\frac{13e}{12}+\frac{4U}{3}
-\frac{U^2}{o-4U/5}.
\end{align}
For $0\le U\le o/\mathrm e$, put $s=U/o$.
Since $s<2/5$,
\begin{align}
\frac{D(U)}o
=
\frac{13e}{12o}
+s\left(\frac43-\frac{s}{1-4s/5}\right)>0.
\end{align}
Indeed, $1-4s/5>0$, and
\begin{align}
\frac43-\frac{s}{1-4s/5}>0
\quad\Longleftrightarrow\quad
s<\frac{20}{31},
\end{align}
which follows from $s<2/5$.

For $U>0$, differentiation gives
\begin{align}
R'(U)
=
\frac{U\bigl(2D(U)-UD'(U)\bigr)}{D(U)^2}.
\end{align}
A direct calculation yields
\begin{align}
2D(U)-UD'(U)
=
\frac{13e}{6}+\frac{4U}{3}
+\frac{(4/5)U^3}{(o-4U/5)^2}>0.
\end{align}
Thus $R'(U)>0$ for $U>0$. Since $R(0)=0$
and $R(U)>0$ for $U>0$, the function is
strictly increasing on $[0,o/\mathrm e]$.
\end{proof}


\subsection{Additional slope bounds and rescaling}

Two further bounds will be used together with
Proposition~\ref{riccati}:
\begin{align}\label{Ucaps}
U\le o-B(r)\le o-d,\qquad
U^2\le H(o^2-d^2).
\end{align}
For the first, nonnegativity of the coefficients
of $A$ gives
\begin{align}
U=rt(B'(r)-A'(r))
\le rtB'(r)
=\sum_k\beta_kkt\exp(-kt).
\end{align}
Therefore
\begin{align}
B(r)+U
&\le
\sum_k\beta_k(1+kt)\exp(-kt)\notag\\
&\le\sum_k\beta_k=o,
\end{align}
where we used $(1+x)\exp(-x)\le1$.
Since $B(r)\ge d$, this proves the first chain
in \eqref{Ucaps}. For the second bound,
multiply \eqref{oddcs} by $r^2t^2$ and use
$U=rtu$ and $H=Mr^2t^2$.

Now let $m>0$ satisfy $mH\ge1$, and define
\begin{align}
Z=\frac{U}{\sqrt{mH}}.
\end{align}
The estimates above imply
\begin{align}\label{Zconditions}
Z\le U,\qquad
d+Z\le o,\qquad
d^2+mZ^2\le o^2,\qquad
\delta\ge R(Z),\qquad
\frac{\mathcal I}{M}=\frac{mQZ^2}{\delta},
\end{align}
where $Q=Q(\lambda,d)$. We verify each implication
to make the role of the rescaling explicit.

First, $mH\ge1$ gives $Z\le U$, so
$U\le o-d$ implies $d+Z\le o$.
Next,
\begin{align}
mZ^2=\frac{U^2}{H}\le o^2-d^2,
\end{align}
which gives the third constraint.
Also $0<Z\le U\le o/\mathrm e$, so
monotonicity of $R$ and \eqref{Rlemma} yield
\begin{align}
\delta\ge R(U)\ge R(Z).
\end{align}

For the information estimate, the definitions give
the exact identity
\begin{align}
mZ^2=\frac{U^2}{H}=\frac{u^2}{M}.
\end{align}
Thus the definition of $\mathcal I$ in \eqref{eq:scalar_bounds_recalled}, together
with $\Phi_\lambda(d)=Q/\delta$, gives the identity
\begin{align}
\frac{\mathcal I}{M}
=\frac{u^2}{M}\frac{Q}{\delta}
=\frac{mQZ^2}{\delta}.
\end{align}
The condition $mH\ge1$ is needed for
$Z\le U$ and the resulting constraints,
but not for this algebraic identity.

Finally, define
\begin{align}
A_*=\frac{13e}{12}+\frac{4Z}{3}.
\end{align}
The bound $\delta\ge R(Z)$ gives
\begin{align}
\frac{Z^2}{\delta}
\le
A_*-\frac{Z^2}{o-4Z/5}
\le A_*.
\end{align}
Combining this with the last identity in
\eqref{Zconditions} yields the simpler estimate
\begin{align}\label{coarse}
\frac{\mathcal I}{M}\le mQ A_*.
\end{align}
Section \ref{sec:long} uses this simpler estimate, whereas Section \ref{sec:central} retains the subtracted term.


\section*{Acknowledgment}

This work was supported in part by a grant of access to
OpenAI models through the ChatGPT for Academic Researchers
program. The proof was developed with the assistance of
GPT-6 Astra while investigating a more general inequality
proposed in \cite{chen}. An earlier version of the scalar argument relied on
computer-assisted verification using a numerical
certificate. Through extensive interaction
between the authors and the model, this step was replaced
by an analytical argument that, although lengthy, can be
checked by hand. The authors take full responsibility for
the content and correctness of this work.

\end{document}